\documentclass[onecolumn, 11pt,letterpaper]{article}
\usepackage{graphicx}     			
\usepackage{multirow}               		
\usepackage{amssymb}                	
\usepackage{latexsym}               		
\usepackage{alltt}                  		
\usepackage{amsgen}                 	
\usepackage{amstext}               		
\usepackage{amscd}                  		
\usepackage{algorithmic}
\usepackage{url}
\usepackage{enumerate}
\usepackage{amsthm}
\usepackage{amsmath}
\usepackage{tikz}
\usepackage{float}
\usepackage{enumitem}
\usepackage{mdwlist}
\usepackage{xspace}

\usepackage[left=1in,top=1in,right=1in,bottom = 1in, nohead]{geometry}

\usepackage{bm}                     	
\usepackage{upgreek}              
\usepackage{xspace}    
\usepackage{color}              
\usepackage{colortbl}    
\usepackage{tabularx}
\usepackage{aliascnt}  	
\usepackage{mdframed}
\usepackage{caption}
\usepackage{makecell}
\usepackage{tcolorbox}
\usepackage{float}                 
\usepackage{ifpdf}
    \ifpdf                                                                  
        \usepackage[breaklinks=true]{hyperref}
        \hypersetup{
            colorlinks=false,               
            citebordercolor={0 1 0},        
            urlbordercolor={0 0 1},         
            linkbordercolor={1 0.2 0.2},    
            pdfborder={0 0 1.2},            
        }
    \else
        \usepackage[hypertex, breaklinks=true]{hyperref}       
        \hypersetup{
            colorlinks=true,               
            anchorcolor=yellow,
            linkcolor=green,
            filecolor=green,
            urlcolor=blue,
            citecolor=red,
        }
    \fi
    \hypersetup{                                                            
        pdfpagemode=UseNone,               
        pdfstartview=Fit,
        pdfpagelayout=SinglePage,       
        bookmarks=true,
        bookmarksnumbered=true,
        bookmarksopenlevel=1,            
        bookmarkstype=toc,
        pdftex=true,
        unicode,
        plainpages=false,           
        hypertexnames=true,         
        pdfpagelabels=true,         
        hyperindex=true,
        naturalnames=false,                
    }

  {\refstepcounter{theorem}\trivlist\item       
  [\hskip\labelsep\bf\sf Policy Query \thetheorem]}
  {\endtrivlist}                                

\let\orgautoref\autoref                         		

\renewcommand{\autoref}[1]{
    \def\equationautorefname{Eq.}
    \def\figureautorefname{Fig.}%
    \def\subfigureautorefname{Fig.}%
    \def\lemmaautorefname{Lemma}%
    \def\conjectureautorefname{Conjecture}%
    \def\remarkautorefname{Remark}%
    \def\propositionautorefname{Prop.}%
    \def\corollaryautorefname{Corollary}%
    \def\definitionautorefname{Def.}%
    \def\sectionautorefname{Sect.}%
    \def\subsectionautorefname{Sect.}%
    \def\subsubsectionautorefname{Section}%
    \def\exampleautorefname{Example}%
    \orgautoref{#1}%
}

\newmdtheoremenv{runexample}{Example}
\newcommand{\specificref}[2]{\hyperref[#2]{#1~\ref*{#2}}}

\renewcommand{\(}{\left(}                   

\renewcommand{\epsilon}{\varepsilon}    

\definecolor{gray}{rgb}{0.5,0.5,0.5}
\definecolor{niceblue}{rgb}{.8,.85,1}

\newcommand{\makeop}[2]                         
  {\ifx#2.\def\next##1{}\else\escapechar=-1     
  \def\next##1{\escapechar=92\def#2{#1}}        
  \expandafter\next\expandafter{\string#2}      
  \let\next\makeop\fi\next{#1}}                 
\makeop{{\cal #1}}                              
  \W  .             %

\def \var(#1){{\bf #1}}

\def\AddSpace#1{\ifcat#1a\ \fi#1} 

\newcommand{\silentreminder}[1]{}

\def \up(#1){[#1)}
\def \down(#1){(#1]}

\def \series(#1,#2){#1_1, \dots \; #1_{#2}}
\def \serieszero(#1,#2){#1_0, #1_1, \dots \; #1_{#2}}

\def \para(#1){{\vspace{1ex}\noindent\small\bf #1\hspace{1ex}}}
\def \myem(#1){{\vspace{1ex}\noindent\small\em #1\hspace{1ex}}}

\newcommand{\eat}[1]{}

\renewenvironment{thebibliography}[1]{%
    \begin{oldthebibliography}{#1}%
    	\small
        \setlength{\parskip}{0.1ex}
        \setlength{\itemsep}{0.0ex}%
}%
{%
\end{oldthebibliography}%
}

\def\pflist#1#2{\ifnum \@listdepth >15\relax \@toodeep 
	\else \global\advance\@listdepth\@ne \fi
	\rightmargin \z@ \listparindent\z@ \itemindent\z@
	\csname @list\romannumeral\the\@listdepth\endcsname 
	\def\@itemlabel{#1}\let\makelabel\@mklab \@nmbrlistfalse #2\relax
	\@trivlist
	\parskip\parsep \parindent\listparindent
	\advance\linewidth -\rightmargin \advance\linewidth -\leftmargin
	\advance\@totalleftmargin \leftmargin
	\parshape \@ne \@totalleftmargin \linewidth 
	\ignorespaces}

\usepackage{graphicx}   
\usepackage{hyperref}   
\usepackage{amsmath,amssymb}    
\usepackage{thm-restate}
\usepackage{mdwlist}
\usepackage{xspace}
\usepackage{verbatim}   
\usepackage{xcolor}
\usepackage[mathscr]{euscript}
\usepackage{tikz}
\usepackage{tikz-qtree}
\usepackage{subcaption}
\usetikzlibrary{arrows.meta}
\usetikzlibrary{fit,shapes,trees,shapes.geometric}
\usetikzlibrary{patterns,decorations.pathreplacing,calc}
\usetikzlibrary{arrows}
\usepackage[linesnumbered,algoruled, lined]{algorithm2e}

\usepackage[normalem]{ulem}

\newcommand{\cut}[1]{}   

\newenvironment{packed_enum}{
\begin{enumerate}
   \setlength{\itemsep}{1pt}
  \setlength{\parskip}{0pt}
   \setlength{\parsep}{0pt}
}
{\end{enumerate}}

\newenvironment{packed_grep}{
\begin{description}
   \setlength{\itemsep}{1pt}
   \setlength{\parskip}{0pt}
   \setlength{\parsep}{0pt}
}
{\end{description}}

\newcommand{\ie}{{\em i.e.}\xspace}
\newcommand{\eg}{{\em e.g.}\xspace}

\newcommand{\introparagraph}[1]{\noindent {\bf \em #1.}}  

\newcommand{\setof}[2]{\{{#1}\mid{#2}\}}        

\usepackage{aliascnt} 

\newtheorem{definition}{Definition}

\newtheorem{proposition}{Proposition}
\newtheorem{corollary}{Corollary}
\newtheorem{conjecture}{Conjecture}
\newtheorem{lemma}{Lemma}
\newtheorem{example}{Example}

\newtheorem{theorem}{Theorem}

\providecommand{\corollaryautorefname}{Corollary}
\providecommand{\lemmaautorefname}{Lemma}
\providecommand{\propositionautorefname}{Proposition}
\providecommand{\definitionautorefname}{Definition}

\providecommand{\exampleautorefname}{Example}

\providecommand{\bx}[0]{\mathbf{x}}
\providecommand{\by}[0]{\mathbf{y}}

\providecommand{\bu}[0]{\mathbf{u}}

\providecommand{\ba}[0]{\mathbf{a}}
\providecommand{\bc}[0]{\mathbf{c}}

\providecommand{\hgraph}[0]{\mathcal{H}}
\providecommand{\edges}[0]{\mathcal{E}}
\providecommand{\nodes}[0]{\mathcal{V}}

\providecommand{\slack}[0]{\alpha}

\providecommand{\domain}[0]{\mathbf{dom}}
\providecommand{\dom}[0]{\mathbf{D}}
\providecommand{\tree}[0]{\mathcal{T}}
\providecommand{\dict}[0]{\mathcal{D}}

\providecommand{\delay}[0]{\delta}
\providecommand{\htree}[0]{\mathscr{T}}
\providecommand{\bag}[0]{\mathscr{B}}
\providecommand{\fhw}[1]{\mathsf{fhw}(#1)}

\providecommand{\cmpr}[1]{\textsf{C}_{#1}}
\providecommand{\dmpr}[1]{\textsf{D}_{#1}}

\providecommand{\prq}[2]{{#1}\mid{#2}}
\providecommand{\pangle}[1]{\langle{#1}\rangle}
\providecommand{\pbox}[0]{\mathbf{B}}
\providecommand{\interval}[0]{\mathbf{I}}

\providecommand{\bound}[0]{\mathsf{b}}
\providecommand{\free}[0]{\mathsf{f}}

\hypersetup{
	colorlinks=true,
	citecolor=blue,
	linkcolor=black}
\begin{document}
	
	\title{Compressed Representations of Conjunctive Query Results}
	
	\author{Shaleen Deep$^1$ \and Paraschos Koutris$^1$\\
		\and$^1$University of Wisconsin-Madison, Madison, WI\\ \{shaleen, paris\}@cs.wisc.edu
	}
	
	\maketitle
	
	\begin{abstract}
	Relational queries, and in particular join queries, often generate large output results when
	executed over a huge dataset. In such cases, it is often infeasible to store the
	whole materialized output if we plan to reuse it further down a data
	processing pipeline. Motivated by this problem, we study the construction of 
	space-efficient {\em compressed representations} of the output of conjunctive queries, with the goal of supporting the efficient
	access of the intermediate compressed result for a given access pattern. In particular, we initiate the study of an important tradeoff: minimizing the {\em space}
	necessary to store the compressed result, versus minimizing the {\em answer time} and 
	{\em delay} for an access request over the result.
	Our main contribution is a novel parameterized data structure, which can be tuned to trade off space
	for answer time. The tradeoff allows us to control the space requirement of the
	data structure precisely, and depends both on the structure of the query and the 
	access pattern. We show how we can use the data structure in conjunction with query decomposition techniques in order
	to efficiently represent the outputs for several classes of conjunctive queries.
\end{abstract}

	\section{Introduction}
\label{sec:intro}

In this paper, we study the problem of constructing space-efficient {\em compressed} 
{\em representations} of the output of conjunctive query results, with the goal of 
efficiently supporting a given access pattern directly over the compressed result, 
instead of the original input database. 
In many data management tasks, the data processing pipeline repeatedly accesses 
the result of a conjunctive query (CQ) using a particular access pattern. 
In the simplest case, this access pattern can be to
enumerate the full result (\eg, in a multiquery optimization context). Generally, the
access pattern can specify, or {\em bound}, the values of some variables, and ask to 
enumerate the values of the remaining variables that satisfy the query.

Currently, there are two extremal solutions for this problem. In one extreme, we can materialize
the full result of the CQ and index the result according to the access pattern. However, since
the output result can often be extremely large, storing this index can be prohibitively expensive.
In the other extreme, we can service each access request by executing the CQ directly over the input database
every time. This solution does not need extra storage, but can lead to inefficiencies, since 
computation has to be done from scratch and may be redundant. 
In this work, we explore the design space between these two extremes. In other words,
we want to compress the query output such that it can be stored in a space-efficient way, while we can 
support a given access pattern over the output as fast as possible.
\begin{example}\label{ex:intro}
	Suppose we want to perform an analysis about mutual friends of users in a social network. 
	The friend relation is represented by a symmetric binary relation $R$ of size $N$, where a tuple $R(a,b)$ 
	denotes that user {\em a} is a friend of user {\em b}.  
	The data analysis involves accessing the database through the following pattern: {\em given any two users $x$ and $z$ 
		who are friends, return all mutual friends $y$.} We formalize this task through 
	an {\em adorned view} $V^{\bound \free \bound}(x,y,z) = R(x,y),R(y,z),$ $R(z,x)$. 
	The above formalism says that the view $V$ of the database will be accessed as follows:
	given values for the {\em bound ($\bound$)} variables $x,z$, we have to return the values for the
	{\em free ($\free$)} variable $y$ such that the tuple is in the view $V.$ The sequence $\bound \free \bound$
	is called the {\em access pattern} for the adorned view.
	
	One option to solve this problem is to satisfy each access by evaluating a query on the input database. This 
	approach is space-efficient, since we work directly on the input and need space $O(N)$.
	However, we may potentially have to wait $\Omega(N)$ time to even learn whether 
	there is any returned value for $y$.
	A second option is to materialize the view $V(x,y,z)$ and build a hash index with key $(x,z)$: in this case, we can satisfy any access optimally with constant delay $\tilde{O}(1)$.\footnote{the $\tilde{O}$ notation includes
		a poly-logarithmic dependence on $N$.} On the other hand, the space needed for storing the view  can be $\Omega(N^{3/2})$. 
	
	In this scenario, we would like to construct representations that trade off between space and delay (or answer time). As we will show later, for this particular example we can construct a data structure for any parameter $\tau$ that needs space $O(N^{3/2}/\tau)$, and can answer any access request with
	delay $\tilde{O}(\tau)$. 
\end{example}	

The idea of efficiently compressing query results has recently gained considerable attention,
both in the context of  factorized databases~\cite{DBLP:journals/tods/OlteanuZ15}, as well as
constant-delay enumeration~\cite{Segoufin15,DBLP:conf/csl/BaganDG07}.
In these settings, the focus is to construct compressed representations that allow for enumeration
of the full result with constant delay:  this means that the time between outputting two consecutive tuples is $O(1)$, independent of the size of the data. 
Using factorization techniques, for any input database $D$, we can construct a compressed
data structure for any CQ without projections, called a $d$-representation, using space 
$O(|D|^{\mathsf{fhw}})$, where
$\mathsf{fhw}$ is the {\em fractional hypertree width} of the query~\cite{DBLP:journals/tods/OlteanuZ15}.
Such a $d$-representation guarantees constant delay enumeration of the full result. 
In~\cite{Segoufin13,DBLP:conf/csl/BaganDG07}, the compression of CQs with projections is also
studied, but the setting is restricted to $O(|D|)$ time preprocessing --which also restricts the size of the 
compressed representation to $O(|D|)$.


In this work, we show that we can dramatically decrease the space for the compressed
representation by both $(i)$ taking advantage of the access pattern,
and $(ii)$ tolerating a possibly increased delay.
For instance, a $d$-representation for the query in Example~\ref{ex:intro} needs 
$O(N^{3/2})$ space, while no linear-time preprocessing can support constant delay
enumeration (under reasonable complexity assumptions~\cite{DBLP:conf/csl/BaganDG07}).
However, we show that if we are willing to tolerate a delay of $\tilde{O}(N^{1/2})$, we can
support the access pattern of Example~\ref{ex:intro} using only $\tilde{O}(N)$ space, linear in
the input size.


\smallskip
\introparagraph{Applications} 
We illustrate the applicability of compressed representations of conjunctive queries on two practical problems: 
$(i)$ processing graph queries over relational databases, and 
$(ii)$ scaling up statistical inference engines. 

In the context of graph analytics, the graph to be analyzed is often defined as a {\em declarative} query over a relational schema~\cite{graphgen2017, graphgen2015, graphgen2017adaptive, spartex}. For instance, consider the DBLP dataset, which contains information about which authors write which papers
through a table $R(author, paper)$. To analyze the relationships between co-authors, we will need to extract
the {\em co-author graph}, which we can express as the view $V(x,y) = R(x,p), R(y,p)$. Most graph analytics
algorithms typically access such a graph through an API that asks for the set of neighbors of
a given vertex, which corresponds to the adorned view $V^{\bound \free}(x,y) = R(x,p), R(y,p)$.
Since the option of
materializing the whole graph (here defined as the view $V$) may require prohibitively large space, it is 
desirable to apply techniques that compress $V$, while we can still answer any access request efficiently. Recent
work~\cite{graphgen2017} has proposed compression techniques for this particular domain, but these techniques
are limited to adorned views of the form $V^{\bound \free}(x,y)$, rely on heuristics, and do not provide any formal analysis on the tradeoff between space and runtime.

The second application of query compression is in statistical inference. For example, \textsc{Felix}~\cite{felix}
is an inference engine for Markov Logic Networks over relational data, which provides scalability by
optimizing the access patterns of logical rules that are evaluated during inference.
These access patterns over rules are modeled exactly as adorned views. \textsc{Felix} groups the relations in the 
body of the view in partitions (optimizing for some cost function), and then materializes each partition (which corresponds to materializing a subquery). 
In the one extreme, it will eagerly materialize the whole view, and in the other extreme it will lazily materialize nothing.
The materialization in \textsc{Felix} is discrete, in that it is not possible to partially materialize each subquery.
In contrast, we consider materialization strategies that explore the full continuum between the two extremes.

\smallskip
\introparagraph{Our Contribution} In this work, we study the design space for compressed representations of
conjunctive queries in the full continuum between optimal space and optimal runtime, when our goal is to
optimize for a specific access pattern.

Our main contribution is a novel data structure that {\em $(i)$ can compress the result for every CQ
	without projections according to the access pattern given by an adorned view, and $(ii)$ can be tuned to tradeoff space for delay and answer time.} 
At the one extreme, the data structure achieves constant delay $O(1)$;
At the other extreme it uses linear space $O(|D|)$, but provides a worst delay guarantee. 
Our proposed data structure includes as a special case the data structure developed in~\cite{Cohen2010} 
for the {\em fast set intersection} problem.

To construct our data structure, we need two technical ingredients. 
The first ingredient (Theorem~\ref{thm:main}) is a data structure that trades space with delay 
with respect to the worst-case size bound of the query result. 
As an example of the type of tradeoffs that can be achieved, for any CQ $Q$ without projections and any access pattern, the data structure needs space $\tilde{O}(|D|^{\rho^*}/\tau)$ to achieve delay $\tilde{O}(\tau)$, where $\rho^*$ is the fractional edge cover number of $Q$, and $|D|$ the size of the input database. 
In many cases and for specific access patterns, the data structure can substantially improve upon this tradeoff. To prove Theorem~\ref{thm:main}, we develop novel techniques on how to encode information
about expensive sub-instances of the problem in a balanced way. 

However, Theorem~\ref{thm:main} by its own gives suboptimal tradeoffs, since it ignores structural
properties of the query (for example, for constant delay it materializes the full result). Our second
ingredient (Theorem~\ref{thm:main2}) combines the data structure of Theorem~\ref{thm:main}  with
a type of tree decomposition called {\em connex tree decomposition}~\cite{DBLP:conf/csl/BaganDG07}.
This tree decomposition has the property of restricting the tree structure such that the bound variables in the
adorned view always form a connected component at the top of the tree.

Finally, we discuss the complexity of choosing the optimal parameters for our two main theorems, when
we want to optimize for delay given a space constraint, or vice versa.

\smallskip
\noindent {\bf \em Organization.} We present our framework in Section~\ref{sec:framework}, along with 
the preliminaries and basic notation. Our two main results (Theorems~\ref{thm:main} 
and~\ref{thm:main2}) are presented in Section~\ref{sec:main}. We then present the detailed construction of the data structure of Theorem~\ref{thm:main} in
Section~\ref{sec:primitive}, and of Theorem~\ref{thm:main2} in Section~\ref{sec:fractionalwidth}.
Finally, in Section~\ref{sec:application}, we discuss some complexity results for optimizing the choice
of parameters.

	\section{Problem Setting}
\label{sec:framework}

In this section we present the basic notions and terminology, and then discuss in detail 
our framework.

\subsection{Conjunctive Queries}
\label{subsec:cq}

In this paper we will focus on the class of {\em conjunctive queries} (CQs), which are expressed as
\begin{equation*} \label{eq:q}
Q(\by) = R_1(\bx_1), R_2(\bx_2), \ldots, R_n(\bx_n) 
\end{equation*}
Here, the symbols $\by,\bx_1, \dots, \bx_n$ are vectors that contain {\em variables} 
or {\em constants}, the atom $Q(\by)$ is the {\em head} of the query, and the atoms
$R_1(\bx_1), R_2(\bx_2), \ldots, R_n(\bx_n)$ form the {\em body}.
The variables in the head are a subset of the variables that appear in the body. A CQ is 
{\em full} if every variable in the body appears also in the head, and it is {\em boolean} 
if the head contains no variables, \ie it is of the form $Q()$.
We will typically use the symbols 
$x,y,z,\dots$ to denote variables, and $a,b,c,\dots$ to denote constants.
If $D$ is an input database, we denote by $Q(D)$ the result of running $Q$
over $D$.


\smallskip
\introparagraph{Natural Joins}
If a CQ is full, has no constants and no repeated variables in the same atom, then we
say it is a {\em natural join query}. For instance,
the triangle query $\Delta(x,y,z) = R(x,y), S(y,z), T(z,x)$ is a natural join query.
A natural join can be represented equivalently as a {\em hypergraph} 
$\mathcal{H} = (\nodes, \edges)$, where $\nodes$ is the set of variables, and
for each hyperedge $F \in \edges$ there exists a relation $R_F$ with variables $F$. 
We will write the join as $\Join_{F \in \edges} R_F$.
The size of relation $R_F$ is denoted by $|R_F|$. Given a set of variables $I \subseteq \nodes$, we define $\edges_{I} = \setof{F \in \edges}{F \cap I \neq \emptyset}$.

\smallskip
\introparagraph{Valuations}
A {\em valuation} $v$ over a subset $V$ of the variables is a total function that maps
each variable $x \in V$ to a value $v(x) \in \domain$, where $\domain$ is a domain
of constants. Given a valuation $v$ of the variables $(x_{i_1}, \dots, x_{i_\ell})$, 
we denote $R_F(v) = R_F \ltimes \{ v(x_{i_1}, \dots, x_{i_\ell}) \}$. 

\smallskip
\introparagraph{Join Size Bounds}
Let $\mathcal{H} = (\nodes, \edges)$ be a hypergraph, and $S \subseteq \nodes$.
A weight assignment $\bu = (u_F)_{F \in \edges}$ 
is called a {\em fractional edge cover} of $S$ if 
$(i)$ for every $F \in \edges, u_F \geq 0$  and $(ii)$ for every
$x \in S, \sum_{F:x \in F} u_F \geq 1$. 
The {\em fractional edge cover number} of $S$, denoted by
$\rho^*_{\mathcal{H}}(S)$ is the minimum of
$\sum_{F \in \edges} u_F$ over all fractional edge covers of $S$. 
We write $\rho^*(\mathcal{H}) = \rho^*_{\mathcal{H}}(\nodes)$.

In a celebrated result, Atserias, Grohe and Marx~\cite{AGM} proved that
for every fractional edge cover $\bu$ of $\nodes$,  the size 
of a natural join is bounded using the following inequality, 
known as the {\em AGM inequality}:
\begin{align}\label{eq:agm}
|\Join_{F \in \edges} R_F| \leq \prod_{F \in \edges} |R_F|^{u_F}
\end{align}
The above bound is constructive~\cite{skewstrikesback,ngo2012worst}: 
there exist worst-case algorithms that compute the join $\Join_{F \in \edges} R_F$ in 
time $O(\prod_{F \in \edges} |R_F|^{u_F})$ for every fractional edge cover $\bu$ of $\nodes$.

\smallskip
\introparagraph{Tree Decompositions}
Let $\mathcal{H} = (\nodes, \edges)$ be a hypergraph of a natural join query $Q$.
A {\em tree decomposition} of $\mathcal{H}$ is  a tuple 
$(\htree, (\bag_t)_{t \in V(\htree)})$ where $\htree$ is a tree, and 
every $\bag_t$ is a subset of $\nodes$, called the {\em bag} of $t$, such that 
\begin{packed_enum}
	\item  
	each edge in $\edges$ is contained in some bag $\bag_t$; and
	\item 
	for each $x \in \nodes$, the set of nodes $\{t \mid x \in \bag_t\}$ is connected in $\htree$.
\end{packed_enum}

The {\em fractional hypertree width} of a tree decomposition is 
defined as $\max_{t \in V(\htree)} \rho^*(\bag_t)$, where
$\rho^*(\bag_t)$ is the minimum fractional edge cover of the vertices in $\bag_t$.
The  fractional hypertree width of a query $Q$, denoted $\fhw{Q}$, is the minimum 
fractional hypertree width among all tree decompositions of its hypergraph.

\smallskip
\introparagraph{Computational Model}
To measure the running time of our algorithms, we will use the uniform-cost RAM 
model~\cite{hopcroft1975design}, where data values as well as pointers to
databases are of constant size. Throughout the paper, all complexity results are 
with respect to data complexity (unless explicitly mentioned), 
where the query is assumed fixed. 

We use the notation $\tilde{O}$ to hide a polylogarithmic 
factor $\log^k |D|$ for some constant $k$, where $D$ is the input database.

\subsection{Adorned Views}

In order to model access patterns over a view $Q$ defined over the input database,
we will use the concept of {\em adorned views}~\cite{AdornedViews}.
In an adorned view, each variable in the head of the view definition is associated with a 
binding type, which can be either {\em bound} $(\bound)$ or {\em free} ($\free$).
A view $Q(x_1, \dots, x_k)$ is then written as $Q^{\eta}(x_1, \dots, x_k)$, where
$\eta \in \{\bound,\free\}^k$ is called the {\em access pattern}.
We denote by $\nodes_\bound$ (resp. $\nodes_\free$) the set of bound (resp. free)
variables from $\{x_1, \dots, x_k\}$. 

We can interpret an adorned view as
a function that maps a valuation over the bound variables $\nodes_\bound$ to a
relation over the free variables $\nodes_\free$.
In other words, for each valuation $v$ over $\nodes_\bound$, the adorned view returns the
answer for the query 
$Q^\eta[v] = \{ \nodes_\free \mid Q(x_1, \dots, x_k) \wedge \forall x_i \in \nodes_\bound: x_i = v(x_i)  \}$,
which we will also refer to as an {\em access request}.

\begin{example}
	$\Delta^{\bound \bound \free}(x,y,z) =$  $R(x,y),$ $ S(y,z), T(z,x)$ captures the following access pattern: given values $x=a,y=b$, list all the $z$-values that form a 
	triangle with the edge $R(a,b)$. 
	As another example,  $\Delta^{\free \free \free}(x,y,z) = R(x,y), S(y,z), T(z,x)$
	simply captures the case where we want to perform a full enumeration of all the triangles in the
	result.
	Finally, $\Delta^{\bound}(x) = R(x,y), S(y,z), T(z,x)$ expresses the access pattern where
	given a node with $x=a$, we want to know whether there exists a triangle that contains it or not.
\end{example}   

An adorned view $Q^\eta(x_1, \dots, x_k)$ is {\em boolean} if every head variable is bound, 
it is {\em non-parametric} if every head variable is free, and it is {\em full} if the CQ if full
(\ie, every variable in the body also appears in the head).
Of particular interest is the adorned view that is full and non-parametric, which we call the
{\em full enumeration view}, and simply asks to output the whole result.

\tikzset{%
	>={Latex[width=2mm,length=2mm]},
	base/.style = {rectangle, rounded corners, draw=black,
		minimum width=2cm, minimum height=3cm,
		text centered, font=\sffamily},
	other/.style =  {minimum width=1cm, draw=none,fill=none,text centered, font=\sffamily}                          
}

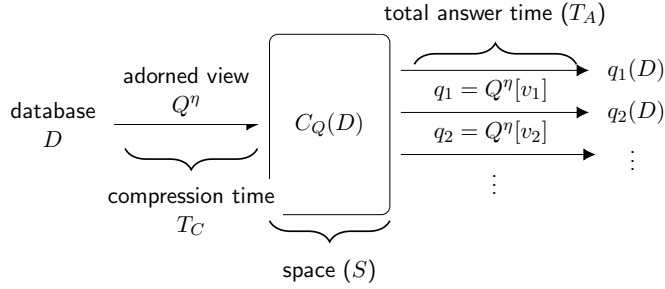
\begin{figure}[t]
	\begin{center}
		\scalebox{0.8}{\begin{tikzpicture}[node distance=0.1cm,
			every node/.style={fill=white, font=\sffamily}, align=center]
			
			\node (input)  [other]       {database \\ $D$};
			\node (compressed)     [base, right of=input, xshift=4.5cm]      { $C_Q(D)$};
			\node (query1)     [other, right of=compressed, xshift=5cm,yshift=-5mm]      {$\vdots$ };
			\node (query2)     [other, below of=query1, yshift=8mm]      { $q_2(D)$};
			\node (query3)     [other, below of=query2, yshift=8mm]      { $q_1(D)$};
			\draw[->,shorten >=5pt, shorten <=5pt ]  (input) -- node [above,midway] {adorned view \\ $Q^\eta$}  (compressed) ;
			\def\brcpad{12mm}   
			\draw [thick, decoration={brace,mirror,raise=4mm,amplitude=10pt},decorate,] 
			($(input)+(\brcpad,0mm)$) to node[below=9mm] {compression time \\ $T_C$} ($(compressed)+(-\brcpad,0mm)$) ;
			\draw [thick, decoration={brace,mirror,raise=16mm,amplitude=10pt},decorate,] 
			($(compressed)+(-10mm,0pt)$) to node[below=21mm] {space ($S$)} ($(compressed)+(10mm,0pt)$) ;
			\draw[->,shorten >=5pt, shorten <=5pt ]  (query1-|compressed.east) -- node [below,midway] { $\vdots$ }  (query1) ;  
			\draw[->,shorten >=5pt, shorten <=5pt ]  (query2-|compressed.east) -- node [below,midway] { $q_2 = Q^\eta[v_2]$ }  (query2) ;  
			\draw[->,shorten >=5pt, shorten <=5pt ]  (query3-|compressed.east) -- node [below,midway] { $q_1 = Q^\eta[v_1]$ }  (query3) ;  
			\draw [thick, decoration={brace,raise=15mm,amplitude=10pt},decorate,] 
			($(query1-|compressed.east)+(4mm,0mm)$) to node[above=20mm] {total answer time ($T_A$)} ($(query1)+(-10mm,0mm)$) ; 
			\end{tikzpicture}}
	\end{center}
	\caption{Depiction of the compression framework along with the parameters.}  
	\label{fig:framework}
\end{figure}

\subsection{Problem Statement}
\label{subsec:cq}

Given an adorned view $Q^\eta(x_1,\dots, x_k)$ and an input database $D$, our goal is to answer any
access request $Q^\eta[v]$ that conforms to the access pattern $\eta$.
The view $Q$ can be expressed through any type of query, but in this work we will focus on the 
case where $Q$ is a conjunctive query. 

There are two extremal approaches to handle this problem. The first solution is to answer any such query 
directly on the input database $D$, without materializing $Q(D)$. This solution is 
efficient in terms of space, but it can lead to inefficient query answering. For instance, consider the
adorned view $\Delta^{\bound \bound \free}(x,y,z) = R(x,y), S(y,z),$ $ T(z,x)$. 
Then, every time we are given new values $x=a, y=b$, we would have to compute all the nodes $c$
that form a triangle with $a,b$, which can be very expensive.

The second solution is to materialize the view $Q(D)$, and then answer
any incoming query over the materialized result. For example, we could choose to materialize
all triangles, and then create an appropriate index over the output result. The drawback of this approach is 
that it requires a lot of space, which may not be available.

We propose to study the solution space between these two extremal solutions, that is, 
instead of materializing all of $Q(D)$, we would like to store a {\em compressed
	representation} $\cmpr{Q}(D)$ of $Q(D)$. The compression function
$\cmpr{Q}$ must guarantee that the compression is lossless, \ie,
there exists a decompression function $\dmpr{Q}$ such that for every
database $D$, it holds that $\dmpr{Q}(\cmpr{Q}(D)) = Q(D)$. We compute the
compressed representation $\cmpr{Q}(D)$
during a {\em preprocessing phase}, and then answer any access request in an 
{\em online phase}.

\smallskip
\introparagraph{Parameters}
Our goal is to construct a compression that is  as space-efficient as possible, 
while it guarantees that we can 
efficiently answer any access query. In particular, we are interested in measuring the
tradeoff between the following parameters, which are also depicted in Figure~\ref{fig:framework}:
\begin{packed_grep}
	\item {\bf Compression Time} ($T_C$): the time to compute 
	$\cmpr{Q}(D)$ during the preprocessing phase. 
	\item {\bf Space} ($S$): the size of $\cmpr{Q}(D)$. 
	\item {\bf Answer Time:} this parameter measures the time to enumerate a query result, where
	the query is of the form $Q^\eta[v]$. The enumeration algorithm must $(i)$ enumerate the 
	query result without any repetitions of tuples, and $(ii)$ use only $O(\log |D|)$ extra memory\footnote{Memory requirement also depends on the memory required for executing the join algorithm. Note that worst case optimal join algorithms such as NPRR~\cite{ngo2012worst} can be executed using $\log |D|$ memory assuming query size is constant and all relations are sorted and indexed.}. 
	We will measure answer time in two different ways.
	\begin{packed_enum}
		\item {\bf delay} ($\delta$): the maximum time to output any two consecutive tuples (and also the 
		time to output the first tuple, and the time to notify that the enumeration has completed). 
		\item {\bf total answer time} ($T_A$): the total time to output the result.
	\end{packed_enum}
\end{packed_grep}

In the case of a boolean adorned view, the delay and the total answer time coincide. In an ideal 
situation, both the compression time and the space are linear to the input size and any query can be answered with constant delay $O(1)$. As we will see later,  
this is achievable in certain cases, but in most cases we have to tradeoff space and 
preprocessing time for delay and total answer time.

\subsection{Some Basic Results}
\label{sec:few}

We present here some basic results that set up a baseline for our framework.
We will study the case where the given view definition $Q$ is a conjunctive query. 

Our first observation is that if we allow the compression time to be at least $\Omega(|D|)$,
we can assume without loss of generality that the adorned view $Q^\eta$
has no constants or repeated variables in a single atom. Indeed, we can first do a 
linear time computation to rewrite the adorned view $Q^\eta$ to a new view where constants
and repeated variables are removed, and then compute the compressed representation
for this new view (with the same adornment). 

\begin{example}
	Consider $Q^{\free \bound}(x,z) = R(x,y,a),S(y,y,z)$. We can
	first compute in linear time $R'(x,y) = R(x,y,a)$ and $S'(y,z) = S(y,y,z)$, and then rewrite the 
	adorned view as $Q^{\free \bound}(x,z) = R'(x,y),S'(y,z)$. 
\end{example}

Hence, whenever the adorned view is a full CQ, we can w.l.o.g. assume that it is  
a natural join query. We now state a simple result for the case where the adorned view is full 
and every variable is bound.

\begin{restatable}{proposition}{linearspacecde}\label{prop:all:bound}
	Suppose that the adorned view is a natural join query with head 
	$Q^{\bound \cdots \bound}(x_1, \dots, x_k)$.
	Then, in time $T_C = O(|D|)$, we can construct a data structure with
	space $S = O(|D|)$, such that we can answer any access request
	over $D$ with constant delay $\delta = O(1)$.
\end{restatable}

Next, consider the full enumeration view $Q^{\free \cdots \free}(x_1,$ $ \dots, x_k)$. 
A first observation is that if we store the materialized view,
we can enumerate the result in constant delay. From the AGM bound, to achieve
this we need space $|D|^{\rho^*(\mathcal{H})}$, where $\mathcal{H}$ is the hypergraph of $Q$.
However, it is possible to improve upon this naive solution using the concept of a
factorized representation~\cite{DBLP:journals/tods/OlteanuZ15}. 
Let $\fhw{Q}$ denote the {\em fractional hypertree width} of $Q$. Then, the result 
from~\cite{DBLP:journals/tods/OlteanuZ15} can be translated in our terminology as follows.

\begin{proposition}[\cite{DBLP:journals/tods/OlteanuZ15}]
	Suppose that the adorned view is a natural join query with head 
	$Q^{\free \cdots \free}(x_1, \dots, x_k)$.
	Then, in compression time $T_C = \tilde{O}(|D|^{\fhw{Q}})$, we can construct a data structure with
	space $S = O(|D|^{\fhw{Q}})$, such that we can answer any access request over $D$ 
	with constant delay $\delta = O(1)$.
\end{proposition}

Since every acyclic query has $\fhw{Q} = 1$, for acyclic CQs without projections both the
compression time and space become linear, $O(|D|)$. In the next section,
we will see how we can generalize the above result to an arbitrary adorned view that is full.

	\section{Main Results and Application}
\label{sec:main}

In this section we present our two main results, and show how they can be applied. 
The first result (Theorem~\ref{thm:main}) is a compression
primitive that can be used with any full adorned view. The second result 
(Theorem~\ref{thm:main2}) builds upon Theorem~\ref{thm:main} and query decomposition
techniques to obtain an improved tradeoff between space and delay.

\subsection{First Main Result}

Consider a full adorned view $Q^\eta(x_1, \dots, x_k)$, where $Q$ is a natural join query expressed by  the hypergraph $\mathcal{H} = (\nodes, \edges)$. Recall that $\nodes_\bound, \nodes_\free$ are the bound and free
variables respectively. Since the query is a natural join and there are no projections, 
we have $\nodes_\bound \cup \nodes_\free = \nodes$. We will denote by
$\mu = |\nodes_\free|$ the number of free variables.
We also impose a lexicographic order on the enumeration order of the output
tuples. Specifically, we equip the domain $\domain$ with a total order $\mathbf{\leq}$,
and then extend this to a total order for output tuples in $\domain^{\mu}$ using some
order $x^1_\free, x^2_\free, \dots, x^\mu_\free$ of the free variables.\footnote{{There is no restriction imposed on the lexicographic ordering of the free variables.}}

\begin{example}\label{ex:running}
	As a running example, consider
	\begin{align*}
	Q^{\free \free \free \bound \bound \bound}(x,y,z,w_1, w_2, w_3)  = & 
	R_1(w_1, x,y), R_2(w_2, y,z), \\ &  R_3(w_3, x,z).
	\end{align*}
	We have
	$\nodes_\free = \{x,y,z\}$ and $\nodes_\bound = \{w_1, w_2, w_3\}$. 
	To keep the exposition simple, assume that $|R_1| = |R_2| = |R_3| = N$.
	
	If we materialize the result and create an index with composite key $(w_1, w_2, w_3)$,
	then in the worst case we need space $S = O(N^3)$, but we will be able to
	enumerate the output for every access request with constant delay. On the other hand, if 
	we create three indexes, one for each $R_i$ with key $w_i$, we can compute
	each access request with worst-case running time and delay of $O(N^{3/2})$. 
	Indeed, once we fix the bound variables to constants $c_1,c_2, c_3$, we need to compute the
	join $R_1(c_1, x,y) \Join R_2(c_2, y,z) \Join R_3(c_3, x,z)$, which needs
	time $O(N^{3/2})$ using any worst-case optimal join algorithm.
\end{example}

For any fractional edge cover $\bu$ of  $\nodes$, and $S \subseteq \nodes$, we define 
the {\em slack} of $\bu$ for $S$ as:
\begin{align} 
\slack(S) = \min_{x \in S} \left( \sum_{F: x \in F} u_F \right)
\end{align}
Intuitively, the slack is the maximum positive quantity such that
$(u_F/\alpha(S))_{F \in \edges}$ is still a fractional edge cover of $S$ . By
construction, the slack is always at least one, $\slack(S) \geq 1$.
For our running example, suppose that we pick a fractional edge cover for 
$\nodes$ with  $u_{R_1} = u_{R_3} =u_{R_3} = 1$. 
Then, the slack of $\bu$ for $\nodes_\free$ is $\slack(\nodes_\free) = 2$.

\begin{theorem} \label{thm:main}
	Let $Q^\eta$ be an adorned view over a natural join query with
	hypergraph $(\nodes, \edges)$. Let $\bu$ be any fractional edge cover of $\nodes$. 
	Then, for any input database $D$ and parameter $\tau >0$ we can construct a data 
	structure with
	\begin{align*}
	\text{compression time } T_C & = \tilde{O} (|D| + \prod_{F \in \edges} |R_F|^{u_F}) \\
	\text{ space } S & = \tilde{O} (|D| + \prod_{F \in \edges} |R_F|^{u_F} / \tau^{\slack(\nodes_\free)})
	\end{align*}
	such that for any access request $q = Q^\eta[v]$, we can enumerate its result
	$q(D)$ in lexicographic order with
	\begin{align*}
	\text{delay } \delta & = \tilde{O}(\tau) \\
	\text{answer time } T_A & = \tilde{O}(|q(D)| + \tau \cdot |q(D)|^{1/\slack(\nodes_\free)})
	\end{align*}
\end{theorem}

\begin{example}
	Let us apply Theorem~\ref{thm:main} to our running example for
	$\bu = (1,1,1)$ and $\tau = N^{1/2}$. The slack for the free variables
	is $\slack(\nodes_\free)=2$.
	The theorem tells us that we can construct in time $\tilde{O}(N^3)$
	a data structure with space $\tilde{O}(N^{2})$, such that every access request $q$ can 
	be answered with delay $\tilde{O}(N^{1/2})$ and answer time
	$\tilde{O}(|q(D)| + \sqrt{N \cdot |q(D)|})$.	 
\end{example}

We prove Theorem~\ref{thm:main} in Section~\ref{sec:primitive}. We 
next show how to apply the theorem to obtain several results on space-efficient 
compressed representations for CQs.

\smallskip
\introparagraph{Applying Theorem~\ref{thm:main}}
We start with the observation that we can always apply Theorem~\ref{thm:main} by
choosing $\bu$ to be the fractional edge cover with optimal value $\rho^*(\mathcal{H})$.
Since the slack is always $\geq 1$, we obtain the following result. 

\begin{proposition} \label{prop:agm}
	Let $Q^\eta$ be an adorned view over a natural join query with
	hypergraph $\mathcal{H}$.
	Then, for any input database $D$ and parameter $\tau > 0$, we can construct a data structure with
	$$ \text{ space } S = \tilde{O} (|D| + |D|^{\rho^*(\mathcal{H})} / \tau)$$
	such that for any access request $q$, we can enumerate its result
	$q(D)$ in lexicographic order with
	\begin{align*}
	\delta  = \tilde{O}(\tau), \quad T_A  = \tilde{O}(\tau \cdot |q(D)|)
	\end{align*}
\end{proposition} 

Proposition~\ref{prop:agm} tells us that the data structure has a linear tradeoff between space 
and delay. Also, to achieve (almost) constant delay $\delta = \tilde{O}(1)$,
the space requirement  becomes $\tilde{O}(|D|^{\rho^*})$; in other words, the
data structure will essentially materialize the whole result. Our second main result will 
allow us to exploit query decomposition to avoid this case.

\begin{example}
	Consider the following adorned view over the Loomis-Whitney join:
	\begin{align*}
	LW_n^{\bound \cdots \bound \free}(x_1, \dots, x_n)  = &  S_1(x_2, \dots, x_n),  S_2(x_1, x_3, \dots, x_n), \\
	& \dots,  S_n(x_1, \dots, x_{n-1})
	\end{align*}
	The minimum fractional edge cover assigns weight $1/(n-1)$
	to each hyperedge and has $\rho^*=n/(n-1)$. Then, 
	Proposition~\ref{prop:agm} tells us that for $\tau > 0$, we can construct a
	compressed representation with space $S = \tilde{O}(|D|+|D|^{n/(n-1)}/\tau)$ and
	delay $\delta = \tilde{O}(\tau)$. 
	Notice that if we aim for linear space, we can choose $\tau = |D|^{1/(n-1)}$
	and achieve a small delay of $\tilde{O}(|D|^{1/(n-1)})$.
\end{example}

Proposition~\ref{prop:agm} ignores the effect of the slack for the free variables. 
The next example shows that taking slack into account is critical in obtaining better tradeoffs.

\begin{example} \label{ex:star}
	Consider the adorned view over the {\em star join}
	$$ S_n^{\bound \cdots \bound \free}(x_1,  \dots, x_n, z) = R_1(x_1,z), R_2(x_2,z), \dots, R_n(x_n, z)$$ 
	The star join is acyclic, which means that the $d$-representation of the full result takes only
	linear space. This $d$-representation can be used for any adornment of $S_n$ where $z$ is a bound
	variable; hence, in all these cases we can guarantee $O(1)$ delay using linear compression space.
	However, we cannot get any guarantees when $z$ is free, as is in the adornment used above.
	
	If we apply Proposition~\ref{prop:agm}, we get space $\tilde{O}(|D|+|D|^n/\tau)$ with delay
	$\tilde{O}(\tau)$. However, we can improve upon this by noticing that for the fractional edge cover
	where $u_1 = \dots = u_n = 1$, the slack is $\slack(\nodes_\free) = n$. 
	Hence, Theorem~\ref{thm:main} 
	tells us that with space $\tilde{O}(|D|^n/\tau^n)$ we get delay $\tilde{O}(\tau)$ and answer time
	$\tilde{O}(|Q(D)| + \tau \cdot |Q(D)|^{1/n})$.
\end{example}

We should note here that our data structure strictly generalizes the data structure proposed in~\cite{Cohen2010}
for  the problem of {\em fast set intersection}. 
Given a family of sets $S_1, \hdots, S_n$, the goal
in this problem is to construct a space-efficient data structure, such that given any two
sets $S_i, S_j$ we can compute their intersection $S_i \cap S_j$ as fast as possible.
It is easy to see that this problem is captured by the adorned view 
$S_2^{\bound \bound \free}(x_1,x_2,z) = R(x_1,z), R(x_2,z)$, where $R$ is a relation that
describes set membership ($R(S_i,a)$ means that $a \in S_i$).

\subsection{Second Main Result}

The direct application of Theorem~\ref{thm:main} can lead to suboptimal tradeoffs 
between space and time/delay, since it ignores the structural properties  of the query. 
In this section, we show how to overcome this problem by combining 
Theorem~\ref{thm:main} with tree decompositions.

We first need to introduce a variant of a tree decomposition of a hypergraph
$\mathcal{H} = (\nodes,\edges)$, defined with respect to a given subset $C \subseteq \nodes$.

\begin{definition}[Connex Tree Decomposition~\cite{DBLP:conf/csl/BaganDG07}] \label{def:fhw}
	Let $\mathcal{H} = (\nodes,\edges)$ be a hypergraph, and $C \subseteq \nodes$.
	A \emph{$C$-connex tree decomposition} of $\mathcal{H}$ 
	is a tuple $(\mathcal{T},A)$, where:
	\begin{packed_enum}
		\item $\mathcal{T}=(\htree, (\bag_t)_{t \in V(\htree)})$ is a tree decomposition of $\mathcal{H}$; and 
		\item 
		$A$ is a connected subset of $V(\htree)$  such that
		$\bigcup_{t \in A} \bag_{t} = C$.
	\end{packed_enum}
\end{definition}

In a $C$-connex tree decomposition, the existence of the set $A$ forces 
the set of nodes that contain some variable from $C$ to be connected in the tree.

\begin{example}\label{ex:path6}
	Consider the hypergraph $\hgraph$ in Figure~\ref{fig:cfhw}.
	The decomposition depicted on the left is a $C$-connex tree decomposition for $C = \emptyset$.
	The $C$-connex tree decomposition on the right is for $C = \{v_1, v_5, v_6\}$. In both
	cases, $A$ consists of a single bag (colored grey) which contains exactly the variables in $C$.
\end{example}

In~\cite{DBLP:conf/csl/BaganDG07}, $C$-connex decompositions were used to obtain 
compressed representations of CQs with projections (where $C$ is the set of the head
variables). In our setting, we will choose $C$ to be the set of bound variables in the adorned
view, \ie, $C = \nodes_\bound$. Additionally, we will use a novel notion of width,
which we introduce next.

Given a $\nodes_\bound$-connex tree decomposition $(\mathcal{T},A)$, 
we orient the tree $\htree$ from some node in $A$. For any node $t \in V(\htree) \setminus A$,
we denote by $\textsf{anc}(t)$ the union of all the bags for the nodes that are the ancestors of $t$.
Define $\nodes_\bound^t = \bag_t \cap \textsf{anc}(t)$ and 
$\nodes_\free^t  = \bag_t \setminus \nodes_\bound^t$. Intuitively, 
$\nodes_\bound^t$ (resp.  $\nodes_\free^t$) are
the bound (resp. free) variables for the bag $t$ as we traverse the tree top-down. 
Figure~\ref{fig:cfhw} depicts each bag $\bag_t$ as $\nodes_\free^t \mid \nodes_\bound^t$.

Given a $\nodes_\bound$-connex tree decomposition, a {\em delay assignment} is a function 
$\delay: V(\htree) \rightarrow [0,\infty)$ that maps each bag to a non-negative number, such
that $\delay(t) = 0$ for $t \in A$. Intuitively, this assignment means that we want to achieve
a delay of $|D|^{\delay(t)}$ for traversing this particular bag.
For a bag $t$, define
\begin{align} \label{eq:rho}
\rho^+_t = \min_{\bu}  \left( \sum_{F} u_F- \delta(t) \cdot \alpha(\nodes_\free^t ) \right)
\end{align}
where $\bu$ is a fractional edge cover of the bag $\bag_t$.
The {\em $\nodes_\bound$-connex fractional hypertree $\delta$-width} of 
$(\mathcal{T},A)$ is defined as $\max_{t \in V(\htree) \setminus A} \rho^+_t $.
It is critical that we ignore the bags in the set $A$ in the max computation. 
We also define $u^+_t = \sum_{F} u'_F$ where $\bu'$ is the fractional edge cover of 
bag $\bag_t$ that minimizes $\rho^+_t$.

When $\delta(t) = 0$ for every bag $t$, the $\delta$-width of any 
$\nodes_\bound$-connex tree decomposition becomes 
$\max_{t \in V(\htree) \setminus A} \rho^*(\bag_t)$, where $\rho^*(\bag_t)$ is the
fractional edge cover number of $\bag_t$.
Define $\textsf{fhw}(\hgraph \mid \nodes_\bound)$ as the smallest such
quantity among all  $\nodes_\bound$-connex tree decompositions of $\mathcal{H}$.
When  $\nodes_\bound = \emptyset$, then
$\textsf{fhw}(\hgraph \mid \nodes_\bound) = \textsf{fhw}(\hgraph)$, thus recovering
the notion of fractional hypertree width. Appendex~\ref{sec:width} shows the relationship between $\textsf{fhw}(\hgraph \mid \nodes_\bound)$ and other hypergraph related parameters.

Finally, we define the $\delay$-height of a $\nodes_\bound$-connex tree decomposition to be 
the maximum weight root-to-leaf path, where the weight of a path $P$ is defined as $\sum_{t \in P} \delay(t)$.

\begin{example}
	Consider the decomposition on the right in Figure~\ref{fig:cfhw}, and a delay assignment $\delay$
	that assigns $1/3$ to node $t_1$ with $\bag_{t_1}= \{v_2,v_4,v_1,v_5\}$, $1/6$ to the bag $t_2$ with $\bag_{t_2} = \{v_2, v_3,v_4\}$, and
	$0$ to the node $t_3$ with $\bag_{t_3} = t_3 = \{v_6,v_7\}$. The $\delta$-height of the tree is $h = \max\{1/3+1/6,0\} = 1/2$. To
	compute the fractional hypertree $\delay$-width, observe that we can cover the bag 
	$\{v_2,v_4,v_1,v_5\}$ by assigning weight of 1 to the edges $\{v_1,v_2\}, \{v_4,v_5\}$, in
	which case $\rho^+_{t_1} = (1+1)-1/3\cdot 1 = 5/3$. We also have
	$\rho^+_{t_2} = (1+1)-1/6\cdot 2 = 5/3$, and $\rho^+_{t_3} = 1$. Hence, the 
	fractional hypertree $\delay$-width is $5/3$. Also, observe that $u^+_{t_1} = u^+_{t_2} = 2$ and $u^+_{t_3} = 1$.\end{example}

\begin{figure}[t]
	\centering
	\scalebox{.9}{\begin{tikzpicture}
		\tikzset{edge/.style = {->,> = latex'},
			vertex/.style={circle, thick, minimum size=5mm}}
		\def\x{0.25}
		
		\begin{scope}[fill opacity=1]
		\node[vertex]  at (0,0) {$v_1$};
		\node[vertex]  at (1,0) {$v_2$};
		\node[vertex]  at (2,0) {$v_3$};
		\node[vertex]  at (3,0) {$v_4$};
		\node[vertex]  at (4,0) {$v_5$};
		\node[vertex]  at (5,0) {$v_6$};
		\node[vertex]  at (6,0) {$v_7$};
		
		\foreach \i in {0,...,5}{
			\draw (\i+\x,0) -- (1+\i-\x,0);}
		
		\draw[fill=black!10] (0,-1) ellipse (0.4cm and 0.33cm) node {};
		\draw[] (0,-2) ellipse (0.8cm and 0.33cm) node {\small $\prq{v_2, v_1}{}$};
		\draw[] (0,-3) ellipse (0.8cm and 0.33cm) node {\small $\prq{v_3}{v_2}$};
		\draw[] (0,-4) ellipse (0.8cm and 0.33cm) node {\small $\prq{v_4}{v_3}$}; 
		\draw[] (0,-5) ellipse (0.8cm and 0.33cm) node {\small $\prq{v_5}{v_4}$};			  			  
		\draw[] (0,-6) ellipse (0.8cm and 0.33cm) node {\small $\prq{v_6}{v_5}$};
		\draw[] (0,-7) ellipse (0.8cm and 0.33cm) node {\small $\prq{v_7}{v_6}$};
		
		\draw[edge] (0,-1.33) -- (0,-1.65);
		\draw[edge] (0,-2.33) -- (0,-2.65);
		\draw[edge] (0,-3.33) -- (0,-3.65);			  
		\draw[edge] (0,-4.33) -- (0,-4.65);
		\draw[edge] (0,-5.33) -- (0,-5.65);
		\draw[edge] (0,-6.33) -- (0,-6.65);
		
		\draw[dotted, line width=0.2mm] (-1,-0.4) -- (7.2, -0.4);
		\draw[dotted, line width=0.2mm] (1.8,-0.4) -- (1.8, -7.2);
		
		\draw[fill=black!10] (5,-2) ellipse (1cm and 0.33cm) node {\small ${\color{red} v_1, v_5, v_6}$};
		\draw[] (3.5,-4) ellipse (1.2cm and 0.33cm) node {\small \small $\prq{v_2, v_4}{\color{red} v_1, v_5}$};
		\draw[] (3.5,-6) ellipse (1cm and 0.33cm) node {\small $\prq{v_3}{v_2 , v_4}$};
		\draw[] (6.5,-4) ellipse (0.8cm and 0.33cm) node {\small $\prq{v_7}{\color{red} v_6}$};
		
		\draw[edge] (5,-2.33) -- (3.5,-3.65);
		\draw[edge] (5,-2.33) -- (6.5,-3.65);
		\draw[edge] (3.5,-4.33) -- (3.5,-5.65);
		\end{scope}	
		\end{tikzpicture}}
	
	\caption{The hypergraph $\hgraph$ for a path query of length $6$, 
		along with two $C$-connex tree decompositions.
		The  decomposition on the left has $C = \emptyset$, and the decomposition on the right  
		$C = \{v_1, v_5, v_6 \}$. The variables in $C$ are colored red, and the grey nodes are the
		ones in the set $A$.}
	\label{fig:cfhw}	
\end{figure}
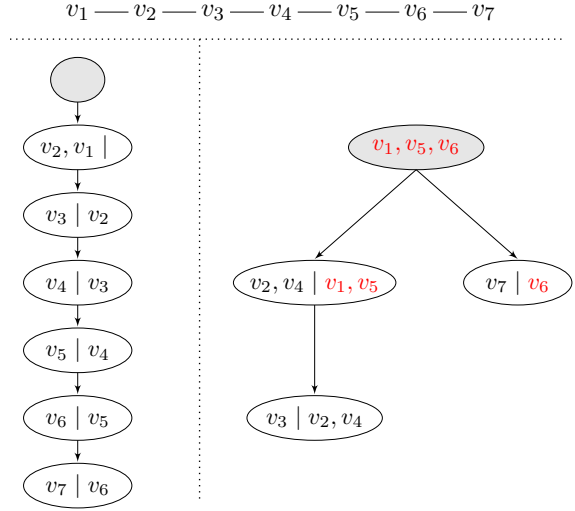

\begin{theorem}\label{thm:main2}
	Let $q = Q^\eta$ be an adorned view over a natural join query with hypergraph $\hgraph = (\nodes, \edges)$. 
	Suppose that $\hgraph$ admits a $\nodes_\bound$-connex tree decomposition.
	Fix any delay assignment $\delay$, and let $f$ be the
	$\nodes_\bound$-connex fractional hypertree $\delay$-width, $h$ the 
	$\delay$-height of the decomposition, and  $u^* = \max_{t \in V(\htree) \setminus A} u^+_t$.
	
	Then, for any input database $D$, we can construct a data structure
	in compression time  $T_C = \tilde{O}(|D| + |D|^{u^* + \max_{t} \delta(t)})$ 
	with space $S = \tilde{O}(|D| + |D|^f)$, such that we can answer any access request 
	with delay $\tilde{O}(|D|^{h})$. 
\end{theorem}

If we write the delay in the above result as $\tilde{O}(\prod_{t \in P} |D|^{\delta(t)})$, where
$P$ is the maximum-weight path, Theorem~\ref{thm:main2} tells us that the delay is
essentially multiplicative in the same branch of the tree, but additive across branches. Unlike Theorem~\ref{thm:main}, the lexicographic ordering of the result $q(D)$ for Theorem~\ref{thm:main2} now depends on the tree decomposition.

For our running example, Theorem~\ref{thm:main2} implies a data structure with
space $\tilde{O}(|D| + |D|^{5/3})$ and delay $\tilde{O}(|D|^{1/2})$. This data structure can be computed in time $\tilde{O}(|D| + |D|^{7/3})$. Notice that this is much smaller than the $O(|D|^4)$ time required to compute the worst case output.
We prove the theorem in detail in Section~\ref{sec:fractionalwidth}, and we discuss
the complexity of choosing the optimal parameters in Section~\ref{sec:application}.
Next, we delve deeper into Theorem~\ref{thm:main2} and how to apply it.

\smallskip
\introparagraph{Applying Theorem~\ref{thm:main2}}
We first give an example where Theorem~\ref{thm:main2} can substantially improve upon the 
space/delay tradeoff of Theorem~\ref{thm:main}.

\begin{example}
	Consider the following adorned view:
	\begin{align*} 
	P_n^{\bound \free \cdots \free \bound} (x_1, \dots, x_{n+1})   =  
	R_1(x_1, x_2), R_2(x_2, x_3), \dots,  R_n(x_n, x_{n+1}).
	\end{align*}
	A direct application of Theorem~\ref{thm:main} results in a tradeoff of space $\tilde{O}(|D| + |D|^{\lceil n/2 \rceil}/\tau)$ with delay
	$\tilde{O}(\tau)$. 
	On the other hand, we can construct a connex tree decomposition where
	$A$ has a single bag $\{x_1,x_{n+1}\}$, which is connected to $\{x_1,x_2, x_n,x_{n+1}\}$,
	which is in turn connected to $\{x_2, x_3, x_{n-1}, x_n\}$, and so on.
	Consider the delay assignment that assigns to each bag $\delta(t) = \log_{|D|} \tau$.
	The $\delta$-width of this decomposition is $2-\log_{|D|} \tau$, while the
	$\delta$-height is $\lfloor n/2 \rfloor \cdot \log_{|D|} \tau $.
	Hence, Theorem~\ref{thm:main2} results in a tradeoff of space $\tilde{O}(|D| + |D|^{2}/\tau)$ with delay
	$\tilde{O}(\tau^{\lfloor n/2 \rfloor})$. 
\end{example}

Suppose now that our goal is to achieve constant delay. From Theorem~\ref{thm:main2}, in 
order to do this we have to choose the delay assignment to be 0 everywhere. In this case,
we have the following result (which slightly strengthens Theorem~\ref{thm:main2}
in this special case by dropping the polylogarithmic dependence).

\begin{proposition}\label{prop:cd}
	Let $Q^\eta$ be a full adorned view over a hypergraph $\hgraph = (\nodes, \edges)$. 
	Then, for any input database $D$, we can construct a data structure in compression time 
	and space $S = O(|D|^{\textsf{fhw}(\hgraph \mid \nodes_\bound)})$, such that we can answer any 
	access request with delay $O(1)$. 
\end{proposition}

Observe that when all variables are free, then $\nodes_\bound = \emptyset$, in which case
$\textsf{fhw}(\hgraph \mid \nodes_\bound) = \textsf{fhw}(\hgraph)$, thus recovering the 
compression result of a $d$-representation. Moreover, since the delay assignment is 0 for all bags, the compression time $T_C = \tilde{O}(|D| + |D|^{\textsf{fhw}(\hgraph \mid \nodes_\bound) })$.

\smallskip
\introparagraph{Beyond full adorned views} Our work provides compression strategies for queries that do not admit out-of-the-box factorization (such as Loomis-Whitney joins), and can also recover the result of compressed $d$-representations as a special case when all variables are free (Proposition~\ref{prop:cd}). On the other hand, factorized databases support a much richer set of queries such as projections, aggregations~\cite{bakibayev2012fdb, bakibayev2013aggregation} and analytical tasks such as learning regression models~\cite{schleich2016learning, olteanu2016factorized}. One possible approach to handling projections in our setting is to force a variable ordering in the $\nodes_{\bound}$-connex decomposition: more precisely, we can force projection variables to appear first in any root to leaf path. This idea of fixing variable ordering would be similar to how order-by clauses are handled in $d$-tree query plans~\cite{bakibayev2013aggregation}. Remarkably, the $\nodes_{\bound}$-connex decomposition in our setting also corresponds to the tree decompositions used to compute aggregations and orderings with group-by attributes as $\nodes_{\bound}$~\cite{olteanu2016factorized}. This points to a deeper connection between our compressed representation and $d$-tree representations used to compute group-by aggregates. We defer the study of these connections and extension of our framework to incorporate more expressive queries 
to future work.

\subsection{A Remark on Optimality} 

So far we have not discussed the
optimality of our results. We remark why proving tight lower bounds might be a hard problem. 

The problem of \textsf{k-SetDisjointness} is defined as follows. Given a family of sets 
$S_1, \hdots, S_m$ of total size $N$, we want to ask queries of the following form: given
as input a subset $I \subseteq \{1, \dots, m\}$ of size $k$, is the intersection $\bigcap_{i \in I} S_i$ empty?
The goal is to construct a space-efficient data structure such that we can answer as fast as possible.
Note that \textsf{k-SetDisjointness} corresponds to the following adorned view:
$Q^{\bound \cdots \bound}(x_1, \dots, x_k) = R(x_1,z), R(x_2,z), \dots, R(x_k, z)$, where $R$ has size $N$.
One can see that we can use the data structure for the corresponding full view with 
head $Q^{\bound \cdots \bound \free}(x_1, \dots, x_k,z)$  (see Example~\ref{ex:star}) to answer 
\textsf{k-SetDisjointness} queries in time $\tilde{O}(\tau)$, using space $\tilde{O}(N^k/\tau^k)$.

In a recent work, Goldstein et al.~\cite{goldstein2017conditional} conjecture the following lower bound:
\begin{conjecture} (due to~\cite{goldstein2017conditional}) 
	\label{conjecture:setintersection:strong}
	Consider a data structure that preprocesses a family of sets $S_1, \hdots, S_m$ of total size $N$.
	If the data structure can answer \textsf{k-SetDisjointness} queries in time (or delay)~\footnote{For boolean queries,
		answer time and delay coincide.} $T$, then it must use
	$S = \tilde{\Omega}(N^k/T^k)$ space.
\end{conjecture}

The above conjecture is a generalization of a conjecture from~\cite{Davoodi12} for the case $k=2$,
which in turn generalizes a folklore conjecture of Patrascu and 
Roditty~\cite{patrascu2010distance}, which was stated only for the case where $\tau=1$ and $k=2$.
Applied in our setting, Conjecture~\ref{conjecture:setintersection:strong} implies that for the adorned view
$Q^{\bound \cdots \bound}(x_1, \dots, x_k) = R_1(x_1,z),$ $R_2(x_2,z), \dots, R_k(x_k, z)$, the tradeoff
between space and delay (or answer time) is essentially optimal when all relations have equal size. Unfortunately, 
proving even the weaker conjecture of~\cite{patrascu2010distance} is considered a hard open problem.

	\section{A Compression Primitive}
\label{sec:primitive}

In this section, we describe the detailed construction of our data structure for Theorem~\ref{thm:main}.

\subsection{Intervals and Boxes}

Before we present the compression procedure, we first introduce two important
concepts in our construction, $\free$-intervals and $\free$-boxes, both of which
describe subspaces of the space of all possible tuples in the output.

\smallskip
\introparagraph{Intervals}
The {\em active domain} $\dom[x]$ of each variable $x$ is equipped with a total order $\leq$
induced from the order of $\domain$.
We will use $\bot, \top$ to denote the smallest and largest element of the 
active domain respectively (these will always exist, since we assume finite databases).
An {\em interval} for variable $x$ is any subset of $\: \dom[x]$ of the form 
$\setof{u \in \dom[x]}{a \leq u \leq b}$, where $a,b \in \dom[x]$, denoted by $[a,b]$. 
We adopt the standard notation for closed and open intervals and write
$[a,b) = \setof{u \in \dom[x]}{a \leq u < b}$, and $(a,b] = \setof{u \in \dom[x]}{a < u \leq b}$. 
The interval $[a,a]$ is called the {\em unit interval} and represents a single value.
We will often write $a$ for the interval $[a,a]$, and the symbol $\square$ for the interval $ \dom[x]$.

By lifting the order from a single domain to the lexicographic order of tuples in 
$\dom_\free = \dom[x^1_\free] \times \dots \times \dom[x^\mu_\free]$, we can also define
intervals over $\dom_\free$, which we call {\em $\free$-intervals}. 
For instance, if $\ba = \pangle{a_1, \dots, a_\mu}$ and $\mathbf{b} = \pangle{b_1, \dots, b_\mu}$,
the $\free$-interval $\interval = [\ba, \mathbf{b})$ represents all valuations $v_\free$ 
over $\nodes_\free$ that
are lexicographically at least $\ba$, but strictly smaller than $\mathbf{b}$.

\smallskip
\introparagraph{Boxes}
It will be useful to consider another type of subsets of
$\dom_\free$, which we call $\free$-boxes.

\begin{definition}[$\free$-box]
An {\em $\free$-box} is defined as a tuple of intervals $\pbox = \pangle{I_1, \dots, I_\mu}$, where
$I_i$ is an interval of $\dom[x_\free\free^i]$. The $\free$-box represents all valuations $v_\free$
over $\nodes_\free$, such that $v_\free(x^i_\free) \in I_i$ for every $i=1, \dots, \mu$.
\end{definition}

We say that a $\free$-box is {\em canonical} if whenever $I_i \neq \square$, then 
every $I_j$ with  $j<i$ is a unit interval. A canonical $\free$-box is always of the form 
$\pangle{a_1, \dots, a_{i-1}, I_i, \square, \dots}$. For ease of notation, we will
omit the $\square$ intervals in the end of a canonical $\free$-box, and simply write
$\pangle{a_1, \dots, a_{i-1}, I_i}$. 

A $\free$-box satisfies the following important property:
%
%

\begin{proposition}\label{prop:box:semijoin}
For every $\free$-box $\pbox$, $(\Join_{F \in \edges} R_F) \ltimes \pbox = \Join_{F \in \edges} (R_F \ltimes \pbox)$.
\end{proposition}

\begin{proof}
Suppose that the $\free$-box is $\pbox = \pangle{I_1, \dots, I_\mu}$. 

Consider some valuation $v$ over $\nodes$ that belongs in $(\Join_{F \in \edges} R_F) \ltimes \pbox $.
Then, for every $F \in \edges$ we have $v(F) \in R_F$, and also for every variable $x_\free^i$ we have
$v(x_\free^i) \in I_i$. Since for every variable in $F \cap \nodes_\free$ we have $v(x_\free^i) \in I_i$ as well,
we conclude that $v(F) \in (R_F \ltimes \pbox$). Thus, $v$ belongs in $\Join_{F \in \edges} (R_F \ltimes \pbox)$
as well.

For the opposite direction, consider some valuation $v$ over $\nodes$ that belongs in $\Join_{F \in \edges} (R_F \ltimes \pbox)$. Since $(R_F \ltimes \pbox) \subseteq R_F$, we have that for every $F \in \edges$,
$v(F) \in R_F$.
Thus, in order to show the desired result, it suffices to show that for every $x_\free^i$ we have
$v(x_\free^i) \in I_i$. Indeed, take any hyperdge $F$ such that $x_\free^i \in F$: then, $v(F) \in (R_F \ltimes \pbox)$ implies that $v(x_\free^i) \in I_i$.
\end{proof}

In other words, if we want to compute the restriction of an output to tuples in $\pbox$,
it suffices to first restrict each relation to $\pbox$ and then perform the join. 
We denote this restriction of the relation as $R_F(\pbox) = R_F \ltimes \pbox$.

Unfortunately, Proposition~\ref{prop:box:semijoin} does not extend to $\free$-intervals.
As we show in the example below, it
is generally not possible to first restrict each relation to $R_F \ltimes \interval$ and then perform
the join.

\begin{example} \label{ex:box}
Consider the adorned view $V^{\free \bound \free \free}(x,y,z,w) = R_1(x,y), R_2(y,z), R_3(z,w), R_4(w,x)$. 
Assume that the active domain is $\dom[x] = \dom[y] = \dom[z] = \dom[w] = \{1,2\}$. Since $\nodes_\free = \{x,z,w\}$, consider the $\free$-interval $\interval = [\ba, \mathbf{b}]$ where $\ba = \pangle{1,2,1}$ and $\mathbf{b} = \pangle{2,1,2}$. In other words, interval $\interval$ contains the following valuations for $\nodes_{\free}$: $(1,2,1), (1,2,2), (2,1,1), (2,1,2)$. It is easy to verify that $R_i \ltimes \interval = R_i$ for every $i=1,2,3,4$ and that $(1,1,1,1)$ is an output tuple. However, $(\Join_{F \in \edges} R_F) \ltimes \interval$ filters out $(1,1,1,1)$ as $(1,1,1)$ does not lie in the interval $\interval$.
\end{example}

As we will see next, we can partition each $\free$-interval to a set
of $\free$-boxes of constant size.

\smallskip
\introparagraph{Box Decomposition}
It will be useful to represent a $\free$-interval $\interval = (\ba, \mathbf{b})$ 
as a union of canonical $\free$-boxes.
Let $j$ be the first position such that $a_j \neq b_j$. Then, we define the {\em box
decomposition} of $\interval$, denoted $\mathcal{B}(\interval)$, as the following set
of canonical $\free$-boxes:
\begin{align*}
\pbox_{\mu}^\ell & = \pangle{a_1, \dots, a_{\mu-1}, (a_\mu, \top]} \displaybreak[0]\\
& \:\dots  \displaybreak[0]\\
\pbox_{j+1}^\ell & = \pangle{a_1, \dots, a_{j}, (a_{j+1}, \top]} \displaybreak[0]\\
\pbox_{j} & = \pangle{a_1, \dots, a_{j-1}, (a_j,b_j)} \displaybreak[0]\\
\pbox_{j+1}^r & = \pangle{b_1, \dots, b_{j}, [\bot, b_{j+1})}  \displaybreak[0]\\
& \:\dots  \displaybreak[0]\\
\pbox_{\mu}^r & = \pangle{b_1, \dots, b_{\mu-1}, [\bot, b_\mu)}
\end{align*}
Intuitively, a box decomposition divides an interval into a set of disjoint, lexicographically ordered intervals. 
We give next an example of an $\free$-interval
and its decomposition into canonical $\free$-boxes.

\begin{example}
	For our running example (Example~\ref{ex:running}), let the active domain be
	$\dom[w_i] = \{1, 2, \dots, 1000 \}$ for $i=1,2,3$.
	Consider an open $\free$-interval $\interval = (\pangle{10,50,100 },\pangle{20,10, 50})$. The box decomposition of $\interval$ consists of the following 5 canonical $\free$-boxes:
	\begin{align*}
		\pbox_{3}^\ell & = \pangle{10, 50, (100, \top]},  \quad
		\pbox_{2}^\ell  = \pangle{10, (50, \top]} \displaybreak[0]\\
		\pbox_{1} & = \pangle{(10,20)}, \displaybreak[0]\\
		\pbox_{2}^r & = \pangle{20, [\bot, 10)} \quad
		\pbox_{3}^r  = \pangle{20, 10, [\bot, 50)}
	\end{align*}
	For another $\free$-interval $\interval' = [\pangle{10,50,100},\pangle{10,50,200})$,
	where the first two positions coincide, 
	the box decomposition consists of one $\free$-box: 
	$\pbox_{3}  = \pangle{10, 50, [100, 200)}$.
\end{example}

The following lemma summarizes several important properties of the box decomposition:
\begin{lemma}\label{lem:box:decomposition}
Let $\interval$ be an $\free$-interval and $\mathcal{B}(\interval)$ be its box
decomposition. Then:
\begin{packed_enum}
\item The $\free$-boxes in $\mathcal{B}(\interval)$ form an order, $\pbox_{\mu}^\ell  \leq \dots \leq \pbox_{j+1}^\ell  \leq \pbox_{j}
\leq \pbox_{j+1}^r \leq \dots \leq \pbox_{\mu}^r$, such that two tuples from different $\free$-boxes are 
ordered according to the order of their $\free$-boxes.
\item The non-empty $\free$-boxes of $\mathcal{B}(\interval)$ form a partition of $\interval$.
\item $|\mathcal{B}(\interval)| \leq 2\mu-1$, where $\mu = |\nodes_\free|$.
\end{packed_enum}
\end{lemma}

\begin{proof}
To show item (1), we begin by considering two consecutive $\free$-boxes of the form $\pbox_i^\ell$.
Consider the largest element $\ba^{>} \in \pbox_{i}^\ell$ and the smallest element $\ba^{<} \in \pbox_{i-1}^\ell$,
for any $i=j+2, \dots, \mu$. Note that $\ba^{>}$ has value $a_{i-1}$ in the $(i-1)$-th position and $\ba^{>}$ has a value from $(a_{i-1}, \top]$ in its $(i-1)$-th position. Since $a_{i-1}$ appears before any element in the set $(a_{i-1}, \top]$ and both boxes agree on the first $i-2$ positions, it follows that $\ba^{>} < \ba^{<}$ (notice that the inequality here is strict). A similar argument applies to all other consecutive $\free$-boxes in the decomposition.

\vspace{2mm}
We next show item (2). We have already shown that the $\free$-boxes in the decomposition are all disjoint. It is 
also easy to observe that every $\free$-box in $\mathcal{B}(\interval)$ is a subset of $\interval$. Thus, in order to
show that the non-empty $\free$-boxes form a partition of $\interval$, it suffices to show that every $\bc \in \interval$
belongs in some $\free$-box of $\mathcal{B}(\interval)$. Let $\interval = (\ba, \mathbf{b})$ and $\bc = \pangle{c_1, \dots, c_\mu}$ such that $\bc \in \interval$. 

We start by looking at the value of $c_j$ where $j$ is the first position such that $a_j \neq b_j$. We 
distinguish three cases. If $a_j < c_j < b_j$, then we have that $\bc \in \pbox_j$ and we are done. 
Suppose now that $c_j = a_j$, and consider the first position $k$ such that $c_k \neq a_k$. 
(Note that  such a $k$ always exists, otherwise $\bc = \ba \not \in \interval$.) Then it is easy to see that
$\bc \in \pbox_k^\ell$. If $c_j = b_j$, then we symmetrically consider the first position $k$ such that
$c_k \neq b_k$; then one can see that $\bc \in \pbox_k^r$.

\vspace{2mm}
To prove item (3), observe that  $|\mathcal{B}(\interval)| = (\mu - j) + 1 + (\mu - j) = 2\mu + (1-2j) \leq 2\mu-1$,
where the last inequality follows because $j \geq 1$.
\end{proof}

The above lemma implies the following corollary:

\begin{corollary} \label{prop:union:box}
Let $\interval$ be an $\free$-interval and $\mathcal{B}(\interval)$ be its box
decomposition.Then:
$$ \bigcup_{\pbox \in \mathcal{B}(\interval)} \Join_{F \in \edges}(R_F \ltimes \pbox) = (\Join_{F \in \edges}R_F) \ltimes \interval$$
\end{corollary}

\subsection{Two Key Ingredients}

We describe here the intuition behind the compression representation. 
Our data structure is parametrized by an integer $\tau \geq 0$, which can be
viewed as a threshold parameter that works as a knob.
We seek to compute the result $(\Join_{F \in \edges} R_F(v_\bound)) \ltimes \interval$, where
$\interval$ is initially the $\free$-interval that represents all possible valuations. 
We can upper bound the running time
for this instance using the AGM bound. If the bound is less than $\tau$,
we can compute the answer in time and delay at most $\tau$. 

Otherwise, we do two things: $(i)$ we store a bit ($\mathsf{1}$ if the answer is nonempty, and $\mathsf{0}$ if it is empty), and $(ii)$ we split the $\free$-interval into two smaller $\free$-intervals.
Then, we recursively apply the same idea for each of the two $\free$-intervals. 
Since we need to store
one bit for every valuation that exceeds the given threshold for a given  $\free$-interval, we need to
bound the number of such valuations: this bound will be our first ingredient. Second,
we split each $\free$-interval in the same way for every valuation; we do it such that
we can balance the information we need to store for each smaller $\free$-interval. The method to
split the $\free$-intervals in a balanced way is our second key ingredient.

\smallskip
\introparagraph{Bounding the Heavy Valuations}
Given a valuation $v_\bound$ for the bound variables, suppose we are asked
to compute the result restricted in some $\free$-interval $\interval$, in other words
$(\Join_{F \in \edges} R_F(v_\bound)) \ltimes \interval$. 
Let $R_F(v,\pbox) = R_F(v) \ltimes \pbox = (R_F \ltimes v) \ltimes \pbox$. 
For an $\free$-box $\pbox$ and valuation $v$ over any variables, we define:
\begin{align*}
T(\pbox) = \prod_{F \in \edges} |R_F(\pbox)|^{\hat u_F}, \quad
T(v, \pbox) = \prod_{F \in \edges} |R_F(v, \pbox)|^{\hat u_F}  
\end{align*}
We overload $T$ to apply to an $\free$-interval $\interval$ and valuation $v$ 
over any variables as follows:
$$ T(\interval) = \sum_{B \in \mathcal{B}(\interval)} T(\pbox),
\quad T(v, \interval) = \sum_{B \in \mathcal{B}(\interval)} T(v,\pbox)$$

\begin{restatable}{proposition}{outputofinterval}
The output $(\Join_{F \in \edges} R_F(v_\bound)) \ltimes \interval$ can be computed
in time $O(T(v_\bound, \interval))$.
\end{restatable}

\begin{proof}
	Consider the box decomposition $\mathcal{B}(\interval)$. First, observe that for any
	$\pbox \in \mathcal{B}(\interval)$ the join $(\Join_{F \in \edges} R_F(v_\bound, \pbox))$
	is over the variables $\nodes_\free$. Since every variable in $\nodes_\free$ is covered by 
	$\hat{\bu}$, we can use any worst-case optimal algorithm to compute the join
	in time at most $T(v_\bound, \pbox)$. By Corollary~\ref{prop:union:box}, we can now compute
	the join over every $\pbox$ and union the (disjoint) results to obtain the desired result. 
	The time needed for this is at most 
	$T(v_\bound, \interval) = \sum_{\pbox \in \mathcal{B}(\interval)} T(v_\bound, \pbox)$. 
\end{proof}

%

We will use the above bound on the running time as a threshold of when it means that a 
particular interval is expensive to compute.

\begin{definition} \label{def:heavy:property}
A pair $(v_\bound, \interval)$ is {\em $\tau$-heavy} for a fractional edge cover $\bu$ if 
$T(v_\bound, \interval) > \tau $.
\end{definition}

Observe that if a pair is not $\tau$-heavy, this means that we can compute the
corresponding subinstance over $\interval$ in time at most $O(\tau)$.
The following proposition provides an upper bound for the number of such $\tau$-heavy pairs.

\begin{restatable}{proposition}{boundedheavyvaluations}
\label{lem:heavy:bound}
Given a $\free$-interval $\interval$ and integer $\tau$, let $\mathcal{H}(\interval,\tau)$ be the valuations $v_\bound$ such that the pair $(v_\bound, \interval)$ is $\tau$-heavy for $\bu$. Then,
$$|\mathcal{H}(\interval,\tau)| \leq \left(\frac{T(\interval)}{\tau} \right)^{\slack} $$
\end{restatable} 

\begin{proof}
	For the sake of simplicity, we will write $\mathcal{H}$ instead of $\mathcal{H}(\interval,\tau)$. We can now write:
	\begin{align*}
		\tau |\mathcal{H}| 
		& \leq \sum_{v_\bound \in \mathcal{H}} \sum_{\pbox \in \mathcal{B}(\interval)}
		\prod_{F \in \edges} |R_F(v_\bound, \pbox)|^{\hat u_F}  \displaybreak[0]\\
		& = \sum_{\pbox \in \mathcal{B}(\interval)}
		\sum_{v_\bound \in \mathcal{H}} 1^{1-1/\slack} \cdot 
		\left( \prod_{F \in \edges} |R_F(v_\bound, \pbox)|^{u_F} \right)^{1/\slack} \displaybreak[0]\\
		& \leq \sum_{\pbox \in \mathcal{B}(\interval)} 
		\left( \sum_{v_\bound \in \mathcal{H}} 1 \right)^{1-1/\slack}
		\left(  \sum_{v_\bound \in \mathcal{H}} \prod_{F \in \edges} |R_F(v_\bound, \pbox)|^{u_F} \right)^{1/\slack} \displaybreak[0]\\
		& = |\mathcal{H}|^{1-1/\slack} \cdot 
		\sum_{\pbox \in \mathcal{B}(\interval)}\left(  \sum_{v_\bound \in \mathcal{H}} \prod_{F \in \edges} |R_F(\pbox) \ltimes v_\bound|^{u_F} \right)^{1/\slack} \displaybreak[0]\\
		& \leq |\mathcal{H}|^{1-1/\slack} \cdot 
		\sum_{\pbox \in \mathcal{B}(\interval)}\left( \prod_{F \in \edges} |R_F(\pbox)|^{u_F} \right)^{1/\slack} \displaybreak[0]\\
		& =   |\mathcal{H}|^{1-1/\slack}  \sum_{\pbox \in \mathcal{B}(\interval)}T(\pbox)
	\end{align*}
	The first inequality comes directly from the definition of a $\tau$-heavy pair. 
	The second inequality is an application of H{\"o}lder's inequality.
	The third inequality is an application of the Query Decomposition Lemma from~\cite{skewstrikesback}.
\end{proof}

\begin{example}
	Consider the following instance for our running example.

	\begin{minipage}[t]{0.3\linewidth}
		\centering
		\begin{tabular}[t]{ !{\vrule width1pt} c|c|c !{\vrule width1pt} } 
			\Xhline{1pt}
			$\mathbf{w_1}$ & $\mathbf{x}$ & $\mathbf{y}$ \\ 
			\Xhline{1pt}
			1 & 1 & 1 \\ 
			1 & 1 & 2 \\
			1 & 2 & 1 \\ 
			2 & 1 & 1 \\
			3 & 1 & 1 \\
			\Xhline{1pt}
		\end{tabular}		
		\vspace{1em}
		
		$R_1$
	\end{minipage}
	\begin{minipage}[t]{0.3\linewidth}
		\centering
		\begin{tabular}[t]{ !{\vrule width1pt} c|c|c !{\vrule width1pt} } 
			\Xhline{1pt}
			$\mathbf{w_2}$ & $\mathbf{y}$ & $\mathbf{z}$ \\ 
			\Xhline{1pt}
			1 & 1 & 2 \\ 
			1 & 2 & 1 \\ 
			1 & 2 & 2 \\ 
			2 & 1 & 1 \\ 
			2 & 1 & 2 \\ 
			\Xhline{1pt}
		\end{tabular}		
		\vspace{1em}
		
		$R_2$
	\end{minipage}
	\begin{minipage}[t]{0.3\linewidth}
		\centering
		\begin{tabular}[t]{ !{\vrule width1pt} c|c|c !{\vrule width1pt} } 
			\Xhline{1pt}
			$\mathbf{w_3}$ & $\mathbf{x}$ & $\mathbf{z}$ \\ 
			\Xhline{1pt}
			1 & 1 & 1 \\ 
			1 & 1 & 2 \\
			1 & 2 & 1 \\ 
			2 & 1 & 1 \\
			2 & 1 & 2 \\
			\Xhline{1pt}
		\end{tabular}		
		\vspace{1em}
		
		$R_3$
	\end{minipage}
	
	We will use $\bu = ( 1, 1 ,1 )$ as the fractional edge cover for $\nodes$. 
	Recall that the slack is $\alpha = 2$, and thus $\hat{\bu} = ( 1/2 ,1/2 ,1/2 )$. 
	Observe that $\dom[x] = \dom[y] = \dom[z] = \{1,2\}$, 
	$\dom[w_1] = \{1,2,3\}$, $\dom[w_2] = \{1,2\}$, $\dom[w_3] = \{1,2,3\}$. Consider the root interval $\interval(r) = [\pangle{1,1,1}, \pangle{2,2,2}]$. The box decomposition $\mathcal{B}(\interval(r))$ is:
	\begin{align*}
		\pbox_{3}^\ell & = \pangle{1, 1, [1,2]},  \quad
		\pbox_{2}^\ell  = \pangle{1, (1, 2]} \displaybreak[0]\\
		\pbox_{2}^r & = \pangle{2, [1,2)} \quad
		\pbox_{3}^r  = \pangle{2, 2, [1, 2]}
	\end{align*}
	We can then compute 
	$T(\interval(r)) =  \sqrt{|3||3||4|} + \sqrt{|1||2||4|} + \sqrt{|1||3||1|} + 0 \approx 10.56$. 
	Consider $v_\bound(w_1, w_2, w_3) = (1,1,1)$. One can compute
	$T(v_\bound, \interval(r)) = \sqrt{2} + 2 + 1 = 4.414$. 
	If we pick  $\tau = 4$, then $(v_\bound, \interval(r))$ is $\tau$-heavy.
\end{example}

\introparagraph{Splitting an Interval}
We next discuss how we perform a balanced splitting of an $\free$-interval $\interval$.

\begin{restatable}{lemma}{friedgut} \label{lem:box:partition}
Let $\pbox = \pangle{I_1, \dots, I_i, \dots}$ be an $\free$-box,
and $J_1, \dots, J_p$ a partition of the interval $I_i$. Denote
$\pbox_k = \pangle{I_1, \dots, J_k, \dots}$. Then, 
$ \sum_{k=1}^p T(\pbox_k) \leq T(\pbox)$.
\end{restatable}

\begin{proof}
	Let $\mathcal{F}$ be the hyperedges that include the variable $x_\free^i$. 
	Notice that if $F \notin \mathcal{F}$, then $R_F(\pbox_k) = R_F(\pbox)$
	for every $k=1, \dots,p $.
	Moreover, observe that for every $F \in \mathcal{F}$, we have
	$\sum_{k=1}^p |R_{F}(\pbox_k)| = |R_F(\pbox)|$.
	Thus, to prove the lemma it suffices to show that
	$$ \sum_{k=1}^p \prod_{F \in \mathcal{F}} |R_{F}(\pbox_k)|^{\hat{u}_F} \leq 
	\prod_{F \in \mathcal{F}} \left( \sum_{k=1}^p |R_{F}(\pbox_k)| \right)^{\hat{u}_F}$$
	The above inequality is an application of Friedgut's inequality~\cite{Friedgut} called the
	generalized H{\"o}lder inequality, which we can apply because  
	$\sum_{F \in \mathcal{F}} \hat{u}_F \geq 1$.
\end{proof}


\begin{restatable}{lemma}{splitpoint} \label{lem:split:interval}
Consider the canonical $\free$-box
$$\pbox = \pangle{a_1, \dots, a_{i-1}, [\beta_L, \beta_U]}.$$
Then, for any $t \geq 0$, there exists $\beta \in \dom[x_\free^i]$ such that 
\begin{packed_enum}
\item $T(\pangle{a_1, \dots, a_{i-1}, [\beta_L, \beta)} \leq t$ 
\item $T(\pangle{a_1, \dots, a_{i-1}, (\beta, \beta_U]}) \leq \max\{0, T(\pbox)-t\}$. 
\end{packed_enum}
Moreover, we can compute $\beta$ in time $\tilde{O}(1)$. 
\end{restatable}

\begin{proof}
	Let $\beta_L = b_1, \dots, b_n = \beta_U$ be the elements of the interval $[\beta_L, \beta_U]$ 
	in sorted order. 
	Define $v_i = T(\pangle{a_1, \dots, a_{i-1}, [\beta_L,b_i]})$ for $i=1, \dots,n$. Observe that 
	we have $v_1 \leq v_2 \leq \dots v_n = T(\pbox)$. 
	Hence, we can view the elements $b_i$ as being sorted in increasing order w.r.t. to 
	the value $v_i$. We now perform binary search to find $\beta = \min_{i} \{v_i  \geq \min (T(\pbox),t) \}$; 
	such an element always exists since $v_i$ is increasing and
	$v_n = T(\pbox)$. We can create an index that returns the count $|R_F(\pbox)|$ in
	logarithmic time, hence the running time to find $\beta$ is $\tilde{O}(1)$.
	By construction, we have
	$T(\pangle{a_1, \dots, a_{i-1}, [\beta_L, \beta)}) \leq \min (T(\pbox),t) \leq t$. 
	Finally, since the intervals $[\beta_L,\beta]$, $[\beta,\beta]$ and
	and $(\beta, \beta_U]$ form a partition of $[\beta_L, \beta_U]$, we can apply 
	Lemma~\ref{lem:box:partition} to obtain that
	$T(\pangle{a_1, \dots, a_{i-1}, (\beta, \beta_U]})  \leq T(\pbox)-\min (T(\pbox),t)
	= \max (0,T(\pbox)-t)$.
\end{proof}

We now present Algorithm~\ref{algo:isplit}, an algorithm that allows for balanced
splitting of an $\free$-interval $\interval$. 

\begin{algorithm}[htp]
  \DontPrintSemicolon
  \LinesNumbered
  \SetNoFillComment
  $\mathcal{B}(\interval) = \{\pbox_1, \dots, \pbox_k \}$ in lexicographic order \;
  $T \gets \sum_{i=1}^k T(\pbox_i)$ \;
  $s \gets \arg \min_{j} \{ \sum_{i=1}^j T(\pbox_i) > T/2\}$  \;
   \BlankLine
 \tcc{let $\pbox_s = \pangle{c_1, \dots, c_{k-1}, I_k, \dots, I_\mu}$}
  
 $\gamma_{k-1} \gets \sum_{i=1}^{s-1} T(\pbox_i), \quad \Delta_{k-1} \gets T(\pbox_s)$ \;
   \BlankLine
 \For{j=k \emph{\KwTo} $\mu$}{
   find min $c_j$ s.t.  
   $T(\pangle{c_1, \dots, c_{j-1}, I_j \cap [\bot,c_j]}) \geq$ 
   $\min\{ \Delta_{j-1}, T/2-\gamma_{j-1}\}$ \;
   $\Delta_j \gets T(\pangle{c_1, \dots, c_{j}})$ \;
   $\gamma_j \gets \gamma_{j-1} +T(\pangle{c_1, \dots, c_{j-1}, I_j \cap [\bot,c_j)}$ \;
 }
 \KwRet{$(c_1, \dots, c_\mu)$}
  \caption{Splitting an $\free$-interval $\interval$}
  \label{algo:isplit}
\end{algorithm}

\begin{proposition}\label{prop:interval:split}
Let $\interval = [\ba, \mathbf{b}]$ be an $\free$-interval. Then, 
Algorithm~\ref{algo:isplit} returns $\mathbf{c} \in \dom_\free$
that splits $\interval$ into $\interval^\prec = [\ba, \mathbf{c})$ and $\interval^\succ = (\mathbf{c}, \mathbf{b}]$
such that
$T(\interval^\prec) \leq T(\interval)/2$ and $T(\interval^\succ) \leq T(\interval)/2$.
Moreover, it terminates in time $\tilde{O}(1)$.
\end{proposition}

\begin{proof}
Notice first that line (6) of the algorithm always finds a $c_j$, following
Lemma~\ref{lem:split:interval}. Hence, the algorithm always returns a 
split point $\mathbf{c} = (c_1, \dots, c_\mu)$.

Define $\pbox^\prec_j = \pangle{c_1, \dots, c_{j-1}, I_j \cap [\bot,c_j)}$ and
$\pbox^\succ_j = \pangle{c_1, \dots, c_{j-1}, I_j \cap (c_j,\top]}$ for 
$j = k, \dots, \mu$.
Similarly to $\gamma_j$, define $\bar{\gamma}_{k-1} = \sum_{i=s+1}^\mu T(\pbox_i)$,
and for $j=k, \dots, \mu$, $\bar{\gamma}_{j} = \bar{\gamma}_{j-1} + T(\pbox^\succ_j)$.

Now, consider the following sets of canonical $\free$-boxes:
\begin{align*}
\mathcal{B}^\prec & =  \pbox_1, \dots, \pbox_{s-1}, \pbox^\prec_{k}, \dots, \pbox^\prec_\mu \\
 \mathcal{B}^\succ & =  \pbox_{1}, \dots, \pbox_{s-1}, \pbox^\succ_{k}, \dots, \pbox^\succ_{\mu}
\end{align*}
The key observation is that $\mathcal{B}^\prec = \mathcal{B}(\interval^\prec )$ and
$\mathcal{B}^\succ = \mathcal{B}(\interval^\succ)$. Moreover, by
construction $\gamma_{\mu} = \sum_{\pbox \in \mathcal{B}^\prec} T(\pbox)$ and also
$\bar \gamma_{\mu}  = \sum_{\pbox \in \mathcal{B}^\succ} T(\pbox)$.
Thus, to prove the statement, it suffices to show that $\gamma_{\mu}, \bar \gamma_{\mu}  \leq T/2$.

We will first show that for any $j=k-1, \dots, \mu: \gamma_{j} \leq T/2$. For $\gamma_{k-1}$ this follows by our choice of $s$.
For some $j \geq k$,  we have $\gamma_j 
= \gamma_{j-1} + T(\pbox^\prec_j) \leq \gamma_{j-1} + \min \{\Delta_{j-1}, T/2-\gamma_{j-1}\} \leq T/2$,
where the first inequality follows from the choice of $c_j$.

Second, we will show by induction that for $j=k-1, \dots, \mu: \bar \gamma_{j} \leq T/2$.
For $\bar \gamma_{k-1}$, we have
$\bar \gamma_{k-1} = T - \sum_{i=1}^s T(\pbox_i) \leq T-T/2 = T/2$. 
Now, let $j \geq k$. We can write:
\begin{align*}
\bar{\gamma}_{j} & = \bar{\gamma}_{j-1} + T(\pbox^\succ_j) \\
& \leq \bar{\gamma}_{j-1} + \max \{0, \Delta_{j-1} - (T/2- \gamma_{j-1}) \} \\
& = \max \{ \bar{\gamma}_{j-1}, (\Delta_{j-1} + \bar{\gamma}_{j-1} + \gamma_{j-1}) -T/2 \} \
\end{align*}
The first inequality follows from item (2) of Lemma~\ref{lem:split:interval}.
By the inductive hypothesis we have $ \bar{\gamma}_{j-1} \leq T/2$. We next show
that
$ \bar{\gamma}_{j} + \gamma_{j} \leq T - \Delta_{j}$.
From Lemma~\ref{lem:box:partition}, it holds for every $j=k, \dots, \mu$:
\begin{align*}
T(\pbox^\prec_j) + T(\pbox^\succ_j) \leq \Delta_{j-1} - \Delta_j
\end{align*}
Using the above inequality, we can write:
\begin{align*}
\bar{\gamma}_{j} + \gamma_{j} 
& = \sum_{i \neq s} T(\pbox_i)  + \sum_{i=k}^j  (T(\pbox_i^\prec) + T(\pbox_i^\succ)) \\
& \leq \sum_{i \neq s} T(\pbox_i)  + \sum_{i=k}^j  (\Delta_{j-1} - \Delta_j) \\
& = \sum_{i \neq s} T(\pbox_i)  + \Delta_{k-1} - \Delta_j \\
& = T - \Delta_j 
\end{align*}

The runtime bound of $\tilde{O}(1)$ follows from Lemma~\ref{lem:split:interval},
which tells us that we can compute each $c_j$ (line (6)) in time $\tilde{O}(1)$.
\end{proof}

\subsection{The Basic Structure}

We now have all the necessary pieces to describe how we construct the 
compressed representation. Recall that our data structure is parametrized 
by a threshold parameter $\tau$, and by a weight assignment 
$\bu = (u_F)_{F \in \edges}$ that covers the variables in $\nodes$. 
The construction consists of two steps.

\smallskip
\introparagraph{1) The Delay-Balanced Tree}
In the first step, we construct an annotated binary tree $\tree$.
Each node $w \in V(\tree)$ is annotated
with an $\free$-interval $\interval(w)$ and a value $\beta(w) \in \dom_\free$, 
which is chosen according to Algorithm~\ref{algo:isplit}.
The tree is constructed recursively. 

Initially, we create a {\em root} $r$ with interval $\interval(r) = \dom_\free$.
Let $w$ be a node at level $\ell$ with interval $\interval(w) = [\ba, \mathbf{c}]$,
and define the threshold at level $\ell$ to be $\tau_\ell = \tau / 2^{\ell(1-1/\slack)}$.
In the case where $T(\interval(w)) < \tau_\ell$, $w$ is a leaf of the tree. 
Otherwise, using $\beta(w)$ as a splitting point, 
we construct two sub-intervals of $\interval$:
$$\interval^\prec = [\ba, \beta(w)) \text{ and } \: \interval^\succ = (\beta(w), \mathbf{c}].$$
If $\interval^\prec \neq \emptyset$, we create
a new node $w_l$ as the left child of $w$, with interval $\interval(w_l) = \interval^\prec$. 
Similarly, if $\interval^\succ \neq \emptyset$, we create a new node $w_r$ as the right 
child of $w$, with interval $\interval(w_r) = \interval^\succ$.
We call the resulting tree $\tree$ a {\em delay-balanced tree}. 
 
\begin{restatable}{lemma}{delaybalancedtreesize}
Let $\tree$ be a delay-balanced tree. Then:
\begin{packed_enum}
\item For every node $w \in V(\tree)$ at level $\ell$, we have
$T(\interval(w)) \leq T(\interval(r))/2^\ell$.
\item The depth of $\tree$ is at most $O(\log T)$ and its size at most $O(T)$,
where $T = \prod_{F \in \edges} |R_F|^{u_F}/ \tau^{\slack}$.
\end{packed_enum}
\end{restatable}

\begin{proof} 
	If $w_1$ is a child of $w_2$, then we have that
	$T(\interval(w_1)) \leq T(\interval(w_2))/2$ by Proposition~\ref{prop:interval:split}.
	Item (1) follows by a simple induction on the depth of the tree.
	
	Suppose that $w$ is a node at level $\ell$. From the condition that we use to
	stop expanding a node, we have:
	\begin{align*}
		\tau_\ell & \leq T(\interval(w)) \leq T(\interval(r)) /2^\ell \\
		& \leq (2\mu-1) \cdot \prod_{F \in \edges} |R_F|^{\hat u_F} /2^{\ell}
	\end{align*}
	The bound on the size follows from the fact that the tree is binary.
\end{proof}

\begin{example}
	Continuing our running example, we will construct the delay-balanced tree. Since $\ell = 0$ for root node, $\tau_\ell = \tau$. We begin by finding the split point $\beta(r)$ for root node. We start with unit interval $\interval(r)^\prec = [\pangle{1,1,1}, \pangle{1,1,1}]$ and keep increasing the interval range until the join evaluation cost $T(\interval(r)^\prec) > T(\interval(r)) / 2$. For interval $\interval(r)^\prec = [\pangle{1,1,1}, \pangle{1,1,1}]$, the box decomposition is $\mathcal{B}(\interval(r)^\prec) = \pbox_{3}^\ell = \pangle{1, 1, 1}$. 
	
	The reader can verify that $T(\interval(r)^\prec) = \sqrt{|3||1||2|} \approx 2.44 \text{ units}$ and changing the interval to  $\interval'(r)^\prec = [\pangle{1,1,1}, $ $\pangle{1,1,2}]$ gives $T(\interval'(r)^\prec) = \sqrt{|3||3||4|} > T(\interval(r))/2$. Thus, $\beta(r) = (1,1,2)$ and $\interval(r)^\succ = [\pangle{1,2,1}, \pangle{2,2,2}]$ with $T(\interval(r)^\succ) = \sqrt{|1||2||4|} + \sqrt{|1||3||1|} \approx 4.56$ . 
	
	For the next level $\ell = 1$, the threshold $\tau_\ell = \tau/ \sqrt{2} \approx 2.82$. Since $T(\interval(r)^\prec) \leq 2.82$, it is a leaf node. We recursively split $\interval(r_r) = \interval^\succ(r) = [\pangle{1,2,1}, \pangle{2,2,2}]$ into $\interval^\prec(r_r)$ and $\interval^\succ(r_r)$. Fixing $\beta(r_r) = (1,2,2)$, we get $T(\interval^\prec(r_r)) = \sqrt{|1||2||1|} \approx 1.414$ and $\interval^\succ(r_r) = [\pangle{2,1,1}, $ $\pangle{2,2,2}],$ $ T(\interval^\succ(r_r)) = \sqrt{3}$. Since both worst case running times are smaller than $\tau_2 = \tau/2 = 2$, our tree construction is complete. We demonstrate the final delay-balanced tree $\tree$ in Figure~\ref{fig:runexample}.
	
	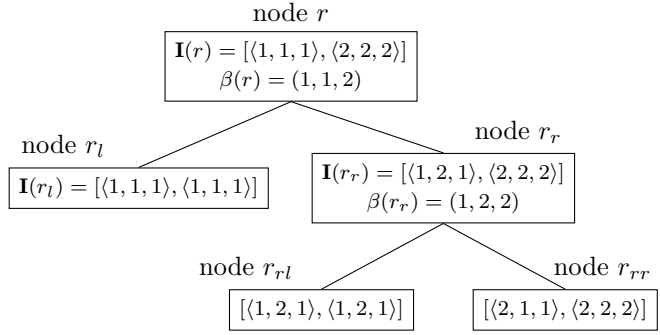
\begin{figure} 
		\centering
		\begin{tikzpicture} 
		[scale=.27,auto,every node/.style={draw},level distance=6cm]
		\tikzstyle{level 1}=[sibling distance = 15cm]
		\tikzstyle{level 2}=[sibling distance = 12cm] 
		\node [label={ \small node $r$}] {\scriptsize 
			\shortstack{$\interval(r) = [\pangle{1,1,1}, \pangle{2,2,2}]$\\$ \beta(r) = (1,1,2)$}
		}
		child {node [label={[xshift=-1.0cm] \small node $r_l$}] {\scriptsize \shortstack{$\interval(r_l) = [\pangle{1,1,1}, \pangle{1,1,1}]$}}}
		child {node [label={[xshift=1cm] \small node $r_r$}] {\scriptsize \shortstack{$\interval(r_r) = [\pangle{1,2,1}, \pangle{2,2,2}]$\\ $ \beta(r_r) = (1,2,2)$}}
			child {node [label={[xshift=-1.0cm] \small node $r_{rl}$}] {\scriptsize \shortstack{$[\pangle{1,2,1}, \pangle{1,2,1}]$}}}
			child {node [ label={[xshift=0.5cm] \small node $r_{rr}$}] {\scriptsize $[\pangle{2,1,1}, \pangle{2,2,2}]$}}
		};
		\end{tikzpicture}
		\caption{Delay balanced tree for running example}
		\label{fig:runexample}	
	\end{figure}

\end{example}

\introparagraph{2) Storing Auxiliary Information}
The second step is to store auxiliary information for the heavy valuations at each node
of the tree $\tree$. Recall that the threshold for a heavy valuation at a node in level $\ell$
is $\tau_\ell = \tau / 2^{\ell (1-1/\slack)}$. We will construct a {\em dictionary} $\dict$ that 
takes as arguments a node $w \in V(\tree)$ at level $\ell$ and a valuation $v_\bound$
such that $(v_\bound, \interval(w))$ is $\tau_\ell$-heavy 
and returns in constant time:
\begin{equation*}
\dict(w, v_\bound) = 
  \begin{cases}
      \mathsf{0}, & \text{ if $(\Join_{F \in \edges} R_F(v_\bound)) \ltimes \interval(w) = \emptyset$}, \\
      \mathsf{1}, & otherwise.
  \end{cases}
\end{equation*}
If $(v_\bound, \interval(w))$ is not $\tau_\ell$-heavy, then there is no entry for this
pair in the dictionary and it simply returns $\bot$. In other words, $\dict$ remembers
for the pairs that are heavy whether the answer is empty or not for the restriction of
the result to the $\free$-interval $\interval(w)$.


We next provide an upper bound on the size of $\dict$.

\begin{restatable}{lemma}{heavyvaluationsbound} \label{lem:dictsizebound}
$|\dict| = \tilde{O}( \prod_{F \in \edges} |R_F|^{u_F} / \tau^{\slack} )$.
\end{restatable}

\begin{proof}
	We first bound the number of $(w, v_\bound)$ pairs that are stored in the dictionary for a 
	node $w$ at level $\ell$. Notice that for node $w$ we will store an entry for at most 
	the $\tau_\ell$-heavy valuations. By Proposition~\ref{lem:heavy:bound}, these are at most
	\begin{align*}
		|\mathcal{H}(\interval(w),\tau_\ell)| & 
		\leq \left( \frac{T(\interval(w))}{ \tau_\ell} \right)^{\slack}
		\leq \left( \frac{T(\interval(r))}{2^\ell \tau_\ell} \right)^{\slack} \\
		& \leq (2\mu-1)^{\slack} 2^{-\ell \slack} \tau_\ell^{-\slack}\prod_{F \in \edges} |R_F|^{u_F}  \displaybreak[0]\\
		& = c \cdot 2^{-\ell} \tau^{-\slack} \prod_{F \in \edges} |R_F|^{u_F}
	\end{align*}
	where $c = (2\mu-1)^{\slack}$ is a constant. 
	At level $\ell$ we have at most $2^\ell$ nodes. Hence, the total number of nodes 
	if the tree has $L$ levels is at most:
	\begin{align*}
		\sum_{\ell=0}^L 2^\ell \left( \tau^{-\slack} 2^{-\ell}  \prod_{F \in \edges} |R_F|^{u_F} \right) 
		\leq \log  |D| \cdot \tau^{-\slack}  \prod_{F \in \edges} |R_F|^{u_F}
	\end{align*}
	This concludes the proof.
\end{proof}

We show in Appendix~\ref{sec:dictionary:construction} a detailed construction that allows us to build
the dictionary $\dict$ in time  $\tilde{O}(\prod_{F \in \edges} |R_F|^{u_F})$, using at most 
$\tilde{O}( \prod_{F \in \edges} |R_F|^{u_F} / \tau^{\slack})$ space, \ie no more space than the size of the
dictionary. 

The final compressed representation consists of the pair $(\tree, \dict)$, along with the necessary indexes on the base relations (that need only linear space).

\begin{example}
	The last step for our running example is to construct the dictionary for all $\tau_\ell$-heavy valuations. Consider the valuation $v_\bound(w_1, w_2, w_3) = (1,1,1)$, which we have shown to be 
	$\tau$-heavy. Next, we store a bit in the dictionary at each node for $v_\bound$ denoting if the join output is non-empty for the restriction of result to interval $\interval$. The reader can verify that 
	$(v_\bound, \interval(r))$  and $(v_\bound, \interval(r_r))$ are $\tau_0$- and $\tau_1$-heavy respectively. Thus, the dictionary will store two entries for $v_\bound$: $\dict(\interval(r), v_\bound) = 1, \dict(\interval(r_r), v_\bound) = 1$.
	
\end{example}

\subsection{Answering a Query}

We now explain how we can use the data structure to answer an access
request $q = Q^\eta[v]$  given by a valuation $v$. 
The detailed algorithm is depicted in Algorithm~\ref{algo:answer}.

\begin{algorithm}[htp]
\SetCommentSty{textsf}
\DontPrintSemicolon 
\SetKwFunction{proc}{\textsf{eval}}
\SetKwInOut{Input}{\textsc{input}}\SetKwInOut{Output}{\textsc{output}}
\Input{tree $\tree$, dictionary $\dict$, valuation $v$}
\Output{query answer $q(D)$}
\BlankLine
%
%
\proc{$r,v_\bound$} \tcc*[r]{start from the root} 
\KwRet{}
\BlankLine
\SetKwProg{myproc}{\textsc{procedure}}{}{}
\myproc{\proc{$w, v_\bound$}}{
   \uIf{$\dict(w, v_\bound) =\bot$ }{
   \ForAll{$\pbox \in \mathcal{B}(\interval(w))$}{
   \textbf{output} $\Join_{F \in \edges} R_F(v_\bound,\pbox)$ }}
    \uElseIf{$\dict(w, v_\bound) =\mathsf{1}$}{
     \uIf{$w$ has left child $w_\ell$}{\proc{$w_\ell, v_\bound$}}
     \textbf{output} $\Join_{F \in \edges} R_F(v_\bound,[\beta(w), \beta(w)]) $ \;
     \uIf{$w$ has right child $w_r$}{\proc{$w_r, v_\bound$}}
    }
    \KwRet{}}
  \caption{Answering a query $q = Q^\eta[v_\bound]$}
  \label{algo:answer}
\end{algorithm}

We start traversing the tree starting from the root $r$.
For a node $w$, if $\dict(w, v_\bound) = \bot$, we
compute the corresponding subinstance using a worst-case optimal algorithm for every box in
the box decomposition. 
If $\dict(w, v_\bound) = \mathsf{0}$, we do nothing. 
If $\dict(w, v_\bound)= \mathsf{1}$, we recursively traverse the left child
(if it exists), compute the instance for the unit interval $[\beta(w), \beta(w)]$, then recursively
traverse the right child (if it exists). This traversal order guarantees that
the tuples are output in lexicographic order.

\smallskip
\introparagraph{Algorithm Analysis}
We now analyze the performance of Algorithm~\ref{algo:answer}.
Let $\tree_{v}$ be the subtree of $\tree$ that contains the nodes visited by
Algorithm~\ref{algo:answer}. The algorithm stops traversing down the tree only when 
it finds a node $w \in V(\tree)$ such that $\dict(w, v_\bound) \neq \textsf{1}$.
(The leaf nodes of $\tree$ have all $\bot$ entries, since by construction they contain
no heavy pairs.)
Thus, the leaves of $\tree_v$ have $\dict(w, v_\bound) \in \{ \textsf{0}, \bot\}$ and
the internal nodes have $\dict(w, v_\bound) = \textsf{1}$.
Figure~\ref{fig:treev} depicts an instance of such an incomplete binary tree.

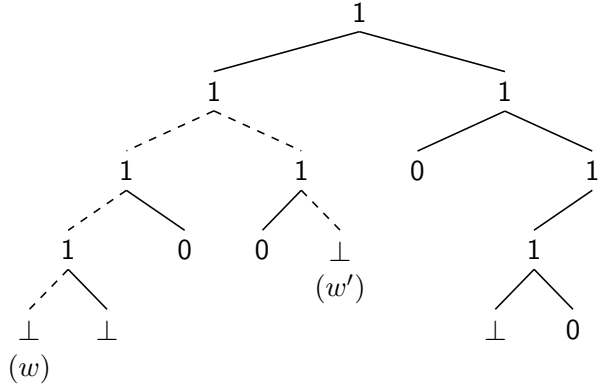
\begin{figure}[t]
\centering
\begin{tikzpicture}
\tikzset{every tree node/.style={minimum width=2.5em, align=center,anchor=north},
blank/.style={draw=none},
edge from parent/.append style={draw=black, line width=0.6},
thick/.style={{very thick}}}
\Tree [.$\mathsf{1}$ [.$\mathsf{1}$ \edge[dashed]; 
  [.$\mathsf{1}$  \edge[dashed]; 
      [.$\mathsf{1}$ \edge[dashed]; $\bot$\\($w$) $\bot$ ] 
      $\mathsf{0}$
      ] 
  \edge[dashed];  [.$\mathsf{1}$ $\mathsf{0}$ \edge[dashed];  $\bot$\\($w'$) ]]
[.$\mathsf{1}$ $\mathsf{0}$
[.$\mathsf{1}$ 
    [.$\mathsf{1}$  $\bot$ $\mathsf{0}$ ] 
    \edge[blank]; \node[blank]{}; ] ] ]
\end{tikzpicture}
\caption{An example subtree $\tree_v$ traversed by Algorithm~\ref{algo:answer}
to answer $q = Q^\eta[v]$. Each node $w$ is annotated by the dictionary entry
$\dict(w, v_\bound)$. The dashed edges show the path from node $w$ that
outputs tuple $t$, to node $w'$ that outputs the lexicographically next tuple $t'$.}
\label{fig:treev}
\end{figure}

\begin{restatable}{lemma}{internalnodetime} \label{lem:internal:constant}
Let $w$ be a node in $\tree_v$. Algorithm~\ref{algo:answer} spends $O(1)$ time at $w$ 
if $\dict(w, v_\bound) \neq \bot$; otherwise it spends time $O(\tau_\ell)$, where
$\ell$ is the level of node $w$.
\end{restatable}

\begin{proof}
	It takes constant time to retrieve the value $\dict(w, v_\bound)$ from the
	dictionary. If the result is $\mathsf{0}$, we do nothing more on node $w$.
	If the result is $\mathsf{1}$, we also need to evaluate the subinstance  
	$\Join_{F \in \edges} R_F(v_\bound, [\beta(w), \beta(w)]) $. But this can be done in 
	constant time, since $[\beta(w), \beta(w)]$ is a unit interval, and thus the evaluation
	can be done by checking a constant number of hash tables.
	
	If $\dict(w, v_\bound) = \bot$ and $w$ is at level $\ell$, the algorithm will
	evaluate the subinstance with interval $\interval(w)$, which by definition
	takes time $O(\tau_\ell)$.
\end{proof}

\begin{proposition}
Algorithm~\ref{algo:answer} enumerates $q(D)$ in lexicographic order
with delay $\delta = \tilde{O}(\tau)$.
\end{proposition}

\begin{proof}
Suppose that Algorithm~\ref{algo:answer} outputs $t \in q(D)$, and the lexicographically
next tuple exists and is $t'$. We will show that the time to output $t'$ after $t$ is 
$\tilde{O}(\tau)$.

If $t,t'$ are output while Algorithm~\ref{algo:answer} is at the same node $w$, it must be
that $\dict(w, v_\bound) = \bot$, in which case the delay will be trivially 
bounded by $O(\tau)$. Otherwise let $w$ be the node where $t$ is output, and $w'$ the
node where $t'$ is output. Notice that $t$ will be the last tuple from $w$ that is output,
and $t'$ the first tuple from $w'$. Now, let $P$ be the unique path in $T_v$ that connects
$w$ with $w'$, with nodes $w=w_1, w_2, \dots, w_k  = w'$. An example of $P$ is
depicted in Figure~\ref{fig:treev}. All the nodes in the path, except possibly the endpoints $w_1,w_k$ are internal nodes and thus we have 
$\dict(w_i, v_\bound) = \mathsf{1}$ for $i=2, \dots, k-1$. Moreover, there must exist 
some $q=1, \dots, k$ such that: $(i)$ if $j \leq q$, then $w_{j-1}$ is a child of $w_{j}$, 
and $(ii)$ if $j > q$, then $w_{j}$ is a child of $w_{j-1}$.

Let us consider the first segment of the path, where $j \leq q$. If $w_{j-1}$ is the right
child of $w_j$, then the algorithm will exit $w_j$ and visit the next node in the path.
If it is the left child, then the algorithm will visit the subtree rooted at its right child first.
However, the subtree can only have a single node $w''$ with $\dict(w''_i, v_\bound) \neq \mathsf{1}$, since otherwise $t'$ would not have been the next tuple to be output. 
Thus, after at most $O(\tau)$ time, the algorithm will visit the next node in the path. 
By a symmetric argument, the algorithm will take at most $O(\tau)$ time to visit the
next node in the path for the second segment, where $j \geq q$. 
Since the length of $P$ is at most 2 times the depth of the tree,
which is $O(\log |D|)$, the algorithm will visit $w'$ (and thus output $t'$) in time
${O}(\tau \log |D|)$.

In the case where there is no next tuple after $t$, it is easy to see that there exists again a
path $P$ that ends at the root node $r$. A similar argument can be done to bound the time
to output the first tuple.
\end{proof}

We now proceed to bound the time to answer the query. The next lemma relates the
output size $|q(D)|$ to the size of the tree $\tree_v$.

\begin{restatable}{lemma}{nodesintree}\label{lem:nodes:tree}
The number of nodes in $\tree_v$ is $ \tilde{O}(|q(D)|)$. 
\end{restatable}

\begin{proof}
	Let $F$ be the set of internal nodes of $\tree_v$, such that there is no child with entry $\mathsf{1}$. The key observation is that $|q(D)| \geq |F|$, since the intervals of
	the nodes in $F$ do not overlap, and each interval will produce at least one output tuple.
	We can easily also see that $|V(\tree_v)| \leq |F| \cdot \log |D|$. Hence, $|V(\tree_v)| = {O}(|q(D)| \cdot \log |D|)$.
\end{proof}

\begin{proposition}
Algorithm~\ref{algo:answer} enumerates $q(D)$ in lexicographic order
in $T_A = \tilde{O}(|q(D)| + \tau \cdot |q(D)|^{1/\slack})$ time.
\end{proposition}

\begin{proof}
We first bound the time needed to visit the nodes $w$ in $\tree_v$ with entry $\neq \bot$.
Since every such node requires constant time to visit, and by Lemma~\ref{lem:nodes:tree} the total number of nodes in tree is $\tilde{O}(|q(D)|)$, we need $\tilde{O}(|q(D)|)$ time. 
Second, we bound the time
to visit the nodes with entry $= \bot$. Let $V$ be the set of such nodes.
Every node in $V$ is a leaf in $\tree_v$. For a node $w$, let $\ell_w$ be its level.
The answer time can be bounded by:
\begin{align*}
\sum_{w \in V} \tau_{\ell_w} & = \sum_{w \in V} \tau \cdot 2^{-\ell_w(1-1/\slack)} \\ 
& = \tau \cdot \sum_{w \in V}  1^{1/\slack} (2^{-\ell_w})^{1-1/\slack} \\
& \leq \tau |V|^{1/\slack} \left( \sum_{w \in V} 2^{-\ell_w} \right)^{1-1/\slack} \\
& \leq \tilde{O}(\tau \cdot |q(D)|^{1/\slack})
\end{align*} 
The first inequality is an application of H{\"o}lders inequality.
The second inequality is an application of Kraft's inequality~\cite{Kraft},
which states that for a binary tree we have 
$\sum_{w \text{ leaf}} 2^{-depth(w)} \leq 1$.
\end{proof}

	\section{Query Decompositions}
\label{sec:fractionalwidth}

In this section, we prove Theorem~\ref{thm:main2}.
Consider an adorned view  over a natural join query (with hypergraph
$\mathcal{H}$), with bound variables $\nodes_\bound$.
Fix a $\nodes_\bound$-connex tree decomposition $(\mathcal{T},A)$. 
Observe that, since the bags in $A$ do not play any role in the width or height, we
can assume w.l.o.g. that $A$ consists of a single bag $t_\bound$.
We consider as the running example in this section the one in Figure~\ref{fig:cfhw}.

\subsection{Constant Delay Enumeration}

As a warm-up, we first show how to construct the data structure for Proposition~\ref{prop:cd},
when our goal is to achieve constant delay for the enumeration of 
every access query.

In the full enumeration case, 
where $\nodes_\bound = \emptyset$, we can take any optimal tree decomposition
with fractional hypertree width $\textsf{fhw}(\mathcal{H})$ (so a $\nodes_\bound$-connex decomposition), 
materialize the tuples in each bag by running the query restricted on the vertices of the bag,
and finally apply a sequence of semi-joins in a bottom-up order  to 
remove any tuples from the bags that do not participate in the final result. Additionally, for each node
$t \in V(\htree)$, we construct a hash index over the materialized result with key 
$\nodes_\bound^t = \bag_t \cap \textsf{anc}(t)$, and output values for the variables in
$\nodes_\free^t = \bag_t \setminus \nodes_\bound^t$.
For example, the node with bag $\{v_2,v_3\}$ constructs an index with key $v_2$
that returns all the matching values of $v_3$.
%
Given such indexes, we can perform a full
enumeration in constant delay starting from the root (which will be the empty bag), and visiting the nodes of 
the tree $\htree$ in a pre-order fashion by following the indexes at each bag. This construction uses the same idea as
$d$-representations~\cite{DBLP:journals/tods/OlteanuZ15}, and requires space
$O(|D|^{\textsf{fhw}(\mathcal{H})})$.

When $\nodes_\bound \neq \emptyset$, the standard tree decomposition may not be useful 
to achieve constant delay enumeration. For instance, for the example hypergraph of Figure~\ref{fig:cfhw}, 
if the adorned view has $\nodes_\bound = \{{v_1, v_5, v_6 }\}$, then the pre-order traversal of the
decomposition on the left will fail to achieve constant delay. 
However, we can use a $\nodes_\bound$-connex decomposition to successfully answer any access request
(\eg, the right decomposition in Figure~\ref{fig:cfhw}). We first materialize all the bags, except for
the bag of $t_\bound$. Then, we run a sequence of semi-joins in a bottom-up manner, where we stop
right before the node $t_\bound$ (since it it not materialized).
For each node $t \in V(\htree) \setminus \{t_\bound\}$, we construct a hash index over the materialized result 
with key $\nodes_\bound^t$. Finally, for the root node $t_\bound$, we construct a hash index that tests
membership for every hyperedge of $\mathcal{H}$ that is contained in $\nodes_\bound$. For the example in Figure~\ref{fig:cfhw},
we construct one such index for the hyperedge $\{v_5,v_6\}$.

Given a valuation $v_\bound$ over $\nodes_\bound$, we answer the access request as follows. 
We start by checking in constant time whether $v_\bound(v_5, v_6)$ is in the input. 
Then, we use the hash index of the node $\{v_2,v_4,v_1,v_5\}$ to find the matching values for $v_2, v_4$
(since $v_1, v_5$ are bound by $v_\bound$), and subsequently use the hash index of $\{v_3,v_2,v_4\}$ to
find the matching values for $v_3$; similarly, we also traverse the right subtree starting of the root node to
find the matching values of $v_7$. We keep doing this traversal until all tuples are enumerated.
We describe next how this algorithm generalizes for any adorned view and beyond constant delay.

%
%

\subsection{Beyond Constant Delay}


We will now sketch the construction and query answering for the general case. 
The detailed construction and algorithms are presented in~\autoref{sec:construction:decomposition},~\ref{sec:answering:decomposition}.
Along with the tree decomposition, let us fix a delay assignment $\delay$.

\smallskip
\introparagraph{Construction Sketch} 
The first step is to apply for each node $t$ in $\htree$ (except $t_\bound$) the construction of the data structure from Theorem~\ref{thm:main}, with the following parameters: 
$(i)$ hypergraph $\mathcal{H'} = (\nodes', \edges')$ where $\nodes' = \bag_t$ and $\edges' = \edges_{\bag_t}$, 
$(ii)$ bound variables $\nodes'_\bound = \nodes_\bound^t$,
$(iii)$ $\tau = |D|^{\delta(t)}$, and $(iv)$ $\bu$ the fractional edge cover that minimizes $\rho_t^+$.
For the root node $t_\bound$, we simply construct a hash index that tests
membership for every hyperedge of $\mathcal{H}$ that is contained in $\nodes_\bound$.
This construction uses for each bag space $\tilde{O}(|D| + |D|^{\rho^+_t - \delta(t)\cdot \alpha(\nodes_\free^t)})$, which means that the compressed representation has size $\tilde{O}(|D| + |D|^{f})$. 

The second step is to set run a sequence of semi-joins in bottom-up fashion. However, since the bags are not fully materialized anymore, this operation is not straightforward. Instead, we set any entry of the dictionary $\dict_t(w, v_\bound)$ of the data structure at node $t$ that is $\mathsf{1}$ to $\mathsf{0}$ if no valuation in
the interval $\interval(w)$ joins with its child bag. This step is necessary to guarantee that if we visit an interval in the delay-balanced tree with entry $\mathsf{1}$, we are certain to produce an output for the full query (and not only the particular bag). To perform this check, we do not actually materialize the bag of the child, 
but we simply use its dictionary (hence costing an extra factor of $\max_t \delta(t)$ during preprocessing time). 

\smallskip
\introparagraph{Query Answering} 
To answer an access query with valuation $v_\bound$,
we start from the root node $t_b$ of the decomposition and check using the indexes of the node
whether $v_\bound$ belongs to all relations $R_F$ such that  $F \subseteq \nodes_\bound$. Then, we invoke Algorithm~\ref{algo:answer} on the leftmost child $t_0$ of $t_b$, which outputs a new valuation
in time at most $\tilde{O}(|D|^{\delay(t_0)})$, or returns nothing.
As soon as we obtain a new output, we 
recursively proceed to the next bag in pre-order traversal of $\htree$, and find valuations for the (still free)
variables in the bag. If there are no such valuations returned by Algorithm~\ref{algo:answer} for the node
under consideration, this means that the last valuation outputted by the parent node does not lead to any output. 
In this case, we resume the enumeration for the parent node. 
Finally, when Algorithm~\ref{algo:answer} finishes the enumeration procedure, then we resume
the enumeration for the pre-order predecessor of the current node (and not the parent).
Intuitively, we go the predecessor to fix our next valuation, in order to enumerate the cartesian product of all free variables in the subtree rooted at the least common ancestor of current node and predecessor node.

The delay guarantee of $\tilde{O}(|D|^{h})$ comes from the fact that, at every node $t$ in the tree,
we will output in time $\tilde{O}(|D|^{\delay(t)})$ at most $\tilde{O}(|D|^{\delay(t)})$ valuations, one
of which will produce a final output tuple. Moreover, when a node has multiple children, then for a fixed
valuation of the node, the traversal of each child is independent of the other children: if
one subtree produces no result, then we can safely exit all subtrees and continue the enumeration of
the node.
The full details and analysis of the algorithm are in~\autoref{sec:answering:decomposition}.

%

	\section{The Complexity of Minimizing Delay}
\label{sec:application}

In this section, we study the computational complexity of choosing the optimal parameters
for Theorem~\ref{thm:main} and Theorem~\ref{thm:main2}. We identify two objectives
that guide the parameter choice: $(i)$ given a space constraint, minimize the delay, and
$(ii)$ given a delay constraint, minimize the necessary space.

We start with the following computational task, which we call \textsc{MinDelayCover}.
We are given as input a full adorned view $Q^\eta$ over a CQ, the sizes $|R_F|$ of each relation $F$, and 
a positive integer $\Sigma$ as a space constraint. 
The size of $Q^\eta$, denoted $\vert Q^\eta \vert$, is defined as the length of $Q^\eta$ when viewed as a word over alphabet that consists of variable set $\nodes$, $\domain$ and atoms in the body of the query. 
The goal is to output a fractional edge cover $\bu$ that minimizes 
the delay in Theorem~\ref{thm:main}, subject to the space constraint $S \leq \Sigma$.


We observe that we can express \textsc{MinDelayCover} as a linear fractional program with a bounded
and non-empty feasible region. Such a program can always be transformed to an equivalent linear
program~\cite{charnes1962programming}, which means that the problem can be solved in polynomial time.  
\begin{proposition} \label{prop:delay}
	\textsc{MinDelayCover} can be solved in polynomial time in the size of the adorned view, the relation
	sizes, and the space constraint.
\end{proposition}
\begin{proof}
	Consider the bilinear program in Figure~\ref{bilinear}. Without loss of
	generality, assume that all relations are of the same size. The first
	constraint ensures that $\vert D \vert ^{\sum u_F}/\tau^\alpha \leq \Sigma$,
	while the fourth constraint encodes the fractional edge covers. However, the
	program is not an LP as $\alpha \log \tau$ is a bilinear constraint. We can
	easily transform it into a linear fractional program as shown in
	Figure~\ref{lfp} where $\hat{\tau} = \alpha \log \tau$. Notice that we can
	replace the objective in program~\ref{bilinear} from $\tau$ to $\log \tau$
	without changing the optimal solution. The key idea is that we can convert the
	linear fractional program to a linear program using the Charnes-Cooper
	transformation~\cite{charnes1962programming} provided that the feasible region
	is bounded and non-empty. Our claim follows from the observation that the
	region is indeed bounded since $u_F \leq 1, \alpha \leq \vert Q \vert,
	\hat{\tau} \leq \vert Q \vert^{2} \log \vert D \vert$ and non-empty as $u_F =
	1, \alpha = 1, \tau = \vert D \vert^{\vert Q \vert}$ is a valid solution.    
\end{proof}

\begin{figure*}[!htp]
\begin{subfigure}{0.4\textwidth} 
	\begin{equation*}
	\begin{array}{ll@{}ll}
	\textbf{minimize}  &  \tau&\\
	\textbf{subject to}& \rho \log \vert D \vert \leq \log
	\vert \Sigma \vert + \alpha \log \tau \label{cons:1} \\
	& \rho = \sum_{F \in \edges} u_F \\
	& \forall x \in \nodes_{\free} :  \sum_{F:x \in F} u_F \geq \slack \\
	& \forall x \in \nodes : \sum_{F:x \in F} u_F \geq 1 \\
	& \forall F \in \edges: 0 \leq u_F \leq 1 \\
	& \alpha \geq 1 
	\end{array}
	\end{equation*}
	\caption{Linear program with bilinear constraint}
	\label{bilinear}
\end{subfigure}
\begin{subfigure}{0.1\textwidth}
	\hfill
\end{subfigure}
\begin{subfigure}{0.4\textwidth} 
	\begin{equation*}
	\begin{array}{ll@{}ll}
	\textbf{minimize}  &  \hat{\tau} / \alpha&\\
	\textbf{subject to}& \rho \log \vert D \vert \leq \log
	\vert \Sigma \vert + \hat{\tau} \\
	& \rho = \sum_{F \in \edges} u_F \\
	& \forall x \in \nodes_{\free} :  \sum_{F:x \in F} u_F \geq \slack \\
	& \forall x \in \nodes : \sum_{F:x \in F} u_F \geq 1 \\
	& \forall F \in \edges: 0 \leq u_F \leq 1 \\
	& \alpha, \hat{\tau} \geq 1 
	\end{array}
	\end{equation*}
	\caption{Transformed linear fractional program}
	\label{lfp}
\end{subfigure}
\caption{Left to right: Bilinear program to minimize delay;
	Equivalent linear fractional program}
\end{figure*}

We also consider the inverse task, called \textsc{MinSpaceCover}: 
given as input a full adorned view $Q^\eta$ over a CQ, 
the sizes $|R_F|$ of each relation $F$, and a positive integer $\Delta$ as a delay constraint, we
want to output a fractional edge cover $\bu$ that minimizes 
the space $S$ in Theorem~\ref{thm:main}, subject to the delay constraint $\tau \leq \Delta$.

To solve \textsc{MinSpaceCover}, observe that we can simply perform a binary search over
the space parameter $S$, from $|D|$ to $|D|^{k}$, where $k$ is the number of atoms in $Q$.
For each space, we then run \textsc{MinDelayCover} and check whether the minimum delay returned
satisfies the delay constraint. 
\begin{proposition}
	\textsc{MinSpaceCover} can be solved in polynomial time in the size of the adorned view, the relation
	sizes, and the delay constraint.
\end{proposition}

We next turn our attention to how to optimize for the parameters in Theorem~\ref{thm:main2}.

Suppose we are given a full adorned view $Q^\eta$ over a CQ, the database size $|D|$, and a
space constraint $\Sigma$, and we want to minimize the delay.
If we are given a fixed $\nodes_\bound$-connex tree decomposition,
then we can compute the optimal delay assignment $\delta$ and optimal fractional edge cover
for each bag as follows: we iterate over every bag in the
tree decomposition, and then solve \textsc{MinDelayCover} for each bag using the space constraint.
It is easy to see that the delay that we obtain in each bag must be the delay of an optimal 
delay assignment. For the inverse task where we are given a $\nodes_\bound$-connex tree decomposition,
a delay constraint, and our goal is to minimize the space, we can apply the same binary search 
technique as in the case of \textsc{MinSpaceCover} (observe the that $\delta$-height is also easily
computable in polynomial time).
In other words, we can compute the optimal parameters for our given objective in polynomial time,
as long as we are provided with a tree decomposition.

In the case where the tree decomposition is not given, then the problem of finding the optimal 
data structure according to Theorem~\ref{thm:main2} becomes intractable. Indeed, we have already
seen that if we want to achieve constant delay $\tau=1$, then the tree decomposition that minimizes
the space $S$ is the one that achieves the $\nodes_\bound$-connex fractional hypertree width,
$\fhw{\mathcal{H} \mid \nodes_\bound}$. Since for $\nodes_\bound = \emptyset$ we have
$\fhw{\mathcal{H} \mid \nodes_\bound} = \fhw{\mathcal{H}}$, and finding the optimal fractional
hypertree width is NP-hard~\cite{gottlob2014treewidth}, finding the optimal tree decomposition
for our setting is also NP-hard.

	\section{Related Work}
\label{sec:related}

There has been a significant amount of literature on data compression; a common application is to apply compression in column-stores~\cite{DBLP:conf/sigmod/AbadiMF06}. 
However, such compression methods typically ignore the logical structure that governs data that is a result of a relational query. The key observation is that we can take advantage of the underlying logical structure in order to design algorithms that can compress the data effectively. This idea has been explored before in the context of {\em factorized databases}~\cite{DBLP:journals/tods/OlteanuZ15}, which can be viewed as a form of logical compression. Our approach builds upon the idea of using query decompositions as a factorized representation, and we show that for certain access
patterns it is possible to go below $\vert D \vert^{\textsf{fhw}}$ space for constant delay enumeration. In addition, our results also allow trading off delay for smaller space requirements of the data structure. A long line of work has also investigated 
the application of  a broader set of queries with projections and aggregations~\cite{bakibayev2012fdb, bakibayev2013aggregation}, as well as learning linear regression models over factorized databases~\cite{schleich2016learning, olteanu2016factorized}. Closely related to our setting is the investigation of join-at-a-time query plans, where at each step, a join over one variable is computed~\cite{ciucanu2015worst}. The intermediate results of these plans are partial factorized representations that compress only a part of the query result. Thus, they can be used to tradeoff
space with delay, albeit in a non-tunable manner.

Our work is also connected to the problem of constant-delay enumeration~\cite{Segoufin13, Segoufin15,DBLP:conf/csl/BaganDG07}: in this case, we want to enumerate a query result with constant delay after a linear time preprocessing step. We can view the linear time preprocessing step as a compression algorithm, which needs space  only $O(|D|)$. It has been shown that the class of {\em connex-free acyclic conjunctive queries} can be enumerated with constant delay after a linear-time preprocessing. Hence, in the case of connex-free acyclic CQs, there exists an optimal compression/decompression algorithm.  However, many classes of widely used queries are not factorizable to linear size, and also can not be enumerated with constant-delay after linear-time preprocessing. Examples in this case are the triangle query $\Delta^{\free \free \free}(x,y,z) = R(x,y), R(y,z), R(z,x)$, or the 2-path query $P_2^{\free \free}(x,y) = R(x,y), R(y,z)$. 

Beyond CQs, related work has also focussed on evaluating signed conjunctive queries~\cite{brault2013pertinence, brault2012negative}. 
CQs that contain both positive and negative atoms allow for tractable enumeration algorithms when they are \textit{free-connex signed-acyclic}~\cite{brault2013pertinence}. Nearby problems include counting the output size $|Q(D)|$ using index structures for enumeration~\cite{durand2011polynomials, durand2015structural}, and enumerating more expressive queries over restricted class of databases~\cite{kazana2013enumeration}.

The problem of finding class of queries that can be maintained in constant time under updates and admit constant delay enumeration is also of considerable interest. Recent work~\cite{cqdelay} considered this particular problem and obtained a dichotomy for self-join free and boolean CQs.  Our work is also related to this problem in that the class of such queries have a specific structure that allow constant delay enumeration.

Query compression is also a central problem in graph analytics. Many applications involve extracting insights from relational databases using graph queries. In such situations, most systems load the relational data in-memory and expand it into a graph representation which can become very dense. Analysis of such graphs is infeasible, as the graph size blows up quickly. Recent work~\cite{graphgen2015, graphgen2017, graphgen2017adaptive} introduced the idea of controlled graph expansion by storing information about high-degree nodes and evaluating acyclic CQs over light sub-instances. 
However, this work is restricted only to binary views (\ie, graphs), and does not offer any formal guarantees on delay or answer time. It also does not allow the compressed representation to grow more than linear in the size of the input.

Finally, we also present a connection to the problem of set intersection.  Set intersection has applications in problems related to document indexing~\cite{Cohen2010, afshani2016data} and proving hardness and bounds for space/approximation tradeoff of distance oracles for graphs~\cite{patrascu2010distance, cohen2010hardness}. Previous work~\cite{Cohen2010} has looked at creating a data structure for fast set intersection reporting and the corresponding boolean version. Our main data structure is a strict generalization of the one from~\cite{Cohen2010}.

	\section{Conclusion}
\label{sec:conclusion}

In this paper we propose a novel and tunable data structure that allows us to compress the result of a conjunctive query so that
we can answer efficiently access requests over the query output.

This work initiates an exciting new direction on studying compression tradeoffs for query results, and thus there are several open problems. The main challenge is to show whether our proposed data structure achieves optimal tradeoffs between the various parameters. Recent work~\cite{afshani2016data} makes it plausible to look for lower bounds in the pointer machine model. A second open problem is to explore how our data structures can be modified to support efficient updates of the base tables. Recent results~\cite{cqdelay} indicate that efficient maintenance of CQ  results under updates is in general a hard problem. This creates a new challenge for designing data structure and algorithms that provide theoretical guarantees and work well in practice. A third challenge is to extend our algorithms to support views with projections:
projections add the additional challenge that we have to deal with duplicate tuples in the output.
Finally, an interesting question is whether it is possible to build our data structure \textit{on-the-fly} without a
preprocessing step.

	\bibliographystyle{abbrv}
	\bibliography{reference}  

	\newpage
	\appendix

\section{Dictionary Construction}
\label{sec:dictionary:construction}

In this section, we show how to efficiently construct the dictionary $\dict$ from
the input database $D$. In particular, we prove the following:

\begin{lemma}
	The dictionary $\dict$ can be constructed in time
	$\tilde{O}(\prod_{F \in \edges} |R_F|^{u_F})$, using at most $\tilde{O}( \prod_{F \in
		\edges} |R_F|^{u_F} / \tau^{\slack})$ space.
\end{lemma}

We construct the dictionary $\dict$ as follows:

\smallskip
\introparagraph{a) Find Heavy Valuations} The first step of the algorithm is to compute the list of
heavy valuations $v_{\bound}$ for any interval $\interval$. 
\begin{proposition} \label{obv:heavy:time}
	Let $\mathcal{L}_{\interval}$ denote the sorted list of all valuations
	of $\nodes_{\bound}$ such that $(v_{\bound}, \interval)$ is
	$\tau$-heavy. Then, $\mathcal{L}_{\interval}$ can be constructed in time
	$\tilde{O}(\sum_{\pbox \in \mathcal{B}(\interval)} \prod_{F \in \mathcal{E}_{\nodes_{\bound}}} |R_F \ltimes \pbox|^{u_F})$ and
	using space at most $O( (T(\interval) / \tau)^{\slack})$.
\end{proposition}
\begin{proof}
	The first observation is that for all heavy $(v_{\bound}, \interval)$ valuations, since $T(v_\bound, \interval) > \tau$, there exists a $\pbox \in \mathcal{B}(\interval)$ such that $R_F(v_{\bound}) \ltimes \pbox$ is non-empty for each $F \in \mathcal{E}_{\nodes_{\bound}}$. This implies that $\pi_{F \cap \nodes_{\bound}} (v_\bound) \in \pi_{F \cap \nodes_{\bound}} (R_{F} \ltimes \pbox)$ (otherwise the relation will be empty and $T(v_\bound, \interval) = 0$). Thus, it is sufficient to compute $\pi_{ \nodes_{\bound}}   ((\Join_{F \in \mathcal{E}_{\nodes_{\bound}}} R_F) \ltimes \interval)$ to find all heavy valuations.

	
	We can construct the list $\mathcal{L}_{\interval}$ by running a worst case join algorithm in time $O(\sum_{\pbox \in \mathcal{B}(\interval)} \prod_{F \in \mathcal{E}_{\nodes_{\bound}}} |R_F \ltimes \pbox|^{u_F})$.  Additionally, as soon as the worst case join algorithm generates an output $v_\bound$, we check if $v_{\bound}$ is $\tau$-heavy in $\tilde{O}(1)$ time. This can be done by using linear sized indexes on base relations to count the number of tuples in each relation $R_{F \in \edges}(v_{\bound}, \interval)$ and using $\bu$ as cover to check whether the execution time is greater than the threshold $\tau$. Since we need only
	heavy valuations, we only retain those in memory.
	Proposition~\ref{lem:heavy:bound} bounds the space requirement of
	$\mathcal{L}_{\interval}$ to at most $O((T(\interval)/ \tau)^{\alpha})$. Sorting $\mathcal{L}_{\interval}$ introduces at most an additional $O(\log |D|)$ factor.
\end{proof}

\introparagraph{b) Using join output to create $\dict$} Consider the delay
balanced tree $\tree$ as constructed in the first step. Without loss of
generality, assume that the tree is full. We bound the time taken to create
$\dict(w, v_{\bound})$ for all nodes $w_{L}$ at some level $L$
(applying the same steps to other levels introduces at most $\log |D|$ factor).
Detailed algorithm is presented below.

\begin{algorithm}[htp]
	\SetCommentSty{textsf}
	\DontPrintSemicolon 
	\SetKwFunction{proc}{\textsf{eval}}
	\SetKwInOut{Input}{\textsc{input}}\SetKwInOut{Output}{\textsc{output}}
	\Input{tree $\tree$}
	\Output{$\dict(w, v_\bound)$ for all leaf nodes}
	\BlankLine
	\ForAll{$w$ in $w_L$}{
		\tcc{Run \textsf{NPRR} on $ (\Join_{F \in \edges_{\nodes_\bound}} R_F) \ltimes \interval(w)$ to compute $\mathcal{L}_{\interval(w)}$}
		\ForAll{$v_\bound \in \mathcal{L}_{\interval(w)}$
			\label{loop:start}}{
			$\dict(w, v_\bound) = 0$ \tcc*[r]{initializing
				$\dict$ with all heavy pairs}
			\label{loop:end}
		}
		\tcc{Run \textsf{NPRR} on $ (\Join_{F \in \edges} R_F) \ltimes \interval(w)$}
		\ForAll{\label{join:nprr} $j \leftarrow$ output tuple from \textsf{NPRR} 
			\tcc*[r]{requires $\log |D|$ main memory }  }{
			{$v_\bound \leftarrow 
				\Pi_{\nodes_{\bound}}(j)$} \label{do:start}\\
			\If(\tcc*[f]{binary search over
				$\mathcal{L}_{\interval(w)}$}) {$v_\bound \in \mathcal{L}_{\interval(w)}$
				\label{val:present} }{
				$\dict(w, v_\bound) = 1$
			}
			\label{do:end}
		}
	}
	\BlankLine

	\caption{Create Dictionary $\dict(w, v_\bound)$ for level $L$ nodes in $\tree$}
	\label{algo:lastlevel}
\end{algorithm}

\smallskip
\introparagraph{Algorithm Analysis} 
We first bound the running time of the algorithm.

\begin{proposition} \label{prop:boxmerging}
	Algorithm~\ref{algo:lastlevel} runs in time $\tilde{O}(\prod_{F \in \edges} |R_F|^{u_F})$.
\end{proposition}
\begin{proof}
	We will first compute the time needed to construct the list $\mathcal{L}_{\interval(w)}$ for all nodes $w$ at level $L$. Proposition~\ref{obv:heavy:time} tells us that to find the heavy valuations for an interval $\interval(w)$ we need time $\tilde{O}(\sum_{\pbox \in \mathcal{B}(\interval(w))} \prod_{F \in \mathcal{E}_{\nodes_{\bound}}} |R_F \ltimes \pbox|^{u_F})$. We will apply Lemma~\ref{lem:box:partition} to show that $\sum_{w \in w_L} \sum_{\pbox \in \mathcal{B}(\interval(w))} \prod_{F \in \edges} |R_F \ltimes \pbox|^{u_F} = O(\prod_{F \in \edges} |R_F|^{u_F})$. Consider all the $\free$-boxes in the box decomposition of $\interval(w), \forall w \in w_L$. All $\free$-boxes that have the first $\mu-1$ variables fixed are of the form $\pbox_{\mu}^k = \pangle{a_1, \dots, a_{\mu - 1}, (a^{k}_\mu, b^{k}_\mu)}$. We apply Lemma~\ref{lem:box:partition} with $i = \mu$ to all such boxes. Thus, $\sum_k T(\pbox_{\mu}^k) \leq T(\pangle{a_1, \dots, a_{\mu - 1}, \square})$.
	
	After this step, all $\free$-boxes have unit interval prefix of length at most $\mu - 1$ and have the domain of $x^{\free}_\mu$ as $\square$. Now, we repeatedly apply lemma ~\ref{lem:box:partition} to all boxes with $i = \mu - 1, \mu - 2, \dots, 1$ sequentially. Each application merges the boxes and fixes the domain of $x^i_{\free} = \square$. The last step merges all $\free$-boxes of the form $\pangle{a}$ to $\interval(r) = \pangle{\square,\dots,\square}$. This gives us,
	
\begin{align*}
	\sum_{w \in w_L} \sum_{\pbox \in \mathcal{B}(\interval(w))} \prod_{F \in \edges} |R_F \ltimes \pbox|^{u_F} &\leq \prod_{F \in \edges} |R_F \ltimes \interval(r)|^{u_F} \\ &= O(\prod_{F \in \edges} |R_F|^{u_F})
\end{align*}
	
	The second step is to bound the running time of the worst case join optimal algorithm to compute $(\Join_{F \in \edges} R_F) \ltimes \interval(w)$ in ~\autoref{join:nprr}. Observe that this join can also be computed in worst case time $\sum_{w \in w_L} \sum_{\pbox \in \mathcal{B}(\interval(w))} \prod_{F \in \edges} |R_F \ltimes \pbox|^{u_F} = O(\prod_{F \in \edges} |R_F|^{u_F})$.
	
	Finally, note that all steps in~\autoref{do:start}-\autoref{do:end} are $\tilde{O}(1)$ operations.
\end{proof}

Next, we analyze the space requirement of
Algorithm~\ref{algo:lastlevel}. 

\begin{proposition}
	Algorithm~\ref{algo:lastlevel}  requires space $O(\prod_{F \in \edges}
	|R_F|^{u_F}/\tau^{\slack})$
\end{proposition}
\begin{proof}
	
	Lines~\ref{loop:start}-\ref{loop:end} take $|\dict|$ amount of space.
	The \textsf{NPRR} algorithm requires $\log |D|$ amount of memory to keep track of
	pointers \footnote{If we have at least $|D|$ memory, i.e, all relations can fit
		in memory, then we require no subsequent I/O's}. Since we are only streaming
	through the join output, there is no additional memory overhead in this step.
	Thus, the bound on memory required follows from
	Proposition~\ref{obv:heavy:time} (bounding the size of $\mathcal{L}_{\interval(r)}$) and
	Lemma~\ref{lem:dictsizebound}.
\end{proof}

\section{Constructing Data Structure for Query Decomposition}
\label{sec:construction:decomposition}

In this section, we present the detailed construction of data structure and time required for Theorem~\ref{thm:main2}. Without loss of generality, we assume that all bound variables are present in a single bag $t_\bound$ in the $\nodes_{\bound}$-connex tree decomposition of the hypergraph $\hgraph$. This can be achieved by simply merging all the bags $t$ with $\bag_t \subseteq \nodes_{\bound}$ into $t_\bound$ which is also designated as the root. Note that the delay assignment for root node is $\delta_{t_\bound} = 0$. Let $\lambda(\htree)$ denote the set of all root to leaf paths in $\htree$, $h$ be the $\delta$-height of the decomposition and $f$ be the $\nodes_{\bound}$-connex fractional hypertree $\delta$-width of the decomposition. We also define the quantity $\bu^* = \max_{t \in V(\htree) \setminus t_\bound} (\sum_F u_F)$ where $u$ is the fractional edge cover for bag $\bag_t$.
We will show that in time $T_C = \tilde{O}(|D| + |D|^{\bu^* + \max_t \delta_{t}})$, we can construct required data structure using space $S = \tilde{O}(|D| + |D|^{f})$ for a given $\nodes_{\bound}$-connex tree decomposition $\htree$ of $\delta$-width $f$. The construction will proceed in two steps:

\smallskip
\introparagraph{(i) Apply Theorem~\ref{thm:main} to decomposition}  We apply theorem~\ref{thm:main} to each bag (except $t_\bound$) in the decomposition with the following parameters: $(i) \mathcal{H}^{t} = (\nodes^{t}, \edges^{t})$ where $\nodes^{t} = \bag_t$ and $\edges^{t} = \edges_{\bag_t}$,  $(ii) \nodes^{t}_\bound = \textsf{anc}(t) \cap \bag_t$ and $\nodes^{t}_{\free} = \bag_t \setminus \textsf{anc}(t)$, and $(iii)$ the edge cover $u$ is the cover of the node $t$ in the decomposition corresponding to $\rho^{+}_t$. Thus, in time $\tilde{O}(|D| + |D|^{\bu^*})$ we can the construct delay-balanced tree $\tree_t$ and the corresponding dictionary $\dict_t$  for each bag other than the root. The space requirement for each bag is no more than $\tilde{O}(|D| + |D|^{f})$. 

However, the dictionary $\dict_{t}$ needs to be modified for each bag since  there can be {\em dangling tuples} in a bag that may participate only in the join output of the bag but not in join output of the branch containing $t$. In the following description, we will use the notation $v^{t}_\bound$ to denote a valuation over variables $\nodes^{t}_\bound$. Note that $v_\bound$ is the valuation over $\nodes_{\bound}$.
\begin{algorithm}[!htp]
	\SetCommentSty{textsf}
	\DontPrintSemicolon 
	\SetKwFunction{proc}{\textsf{fullreducer}}
	\SetKwFunction{enqueue}{enqueue}	
	\SetKwFunction{dequeue}{dequeue}
	\SetKwInOut{Input}{\textsc{input}}\SetKwInOut{Output}{\textsc{output}}
	\SetKwData{Order}{traversaltopdown}
	\SetKwData{Stack}{traversalbottomup}
	\SetKwData{node}{node}
	\SetKwData{parent}{parent}
	\SetKwData{child}{child$_{l}$}		
	\Input{$\nodes_{\bound}$-bound decomposition $\htree$, ${(\tree, \dict})_{t \in V(\htree)}$}
	
	\ForAll{ $t \in \htree \setminus \{t_\bound \cup \text{ children of } t_\bound \} \text{ in post-order fashion }$ }{
		\parent$\leftarrow$ \textit{parent of $t$} \;
		\ForAll{$ w \in w_L$ of $\tree_{\parent}$ \tcc*[r]{$w_L$ represents all nodes at level $L$ in $\tree_{\parent}$ }}{
			
			\ForAll{heavy $v^{\parent}_\bound \in w$ and $\dict_{\parent}(w,v^{\parent}_\bound) = 1$}{
				
				\ForAll{$k \leftarrow Q_{\parent}(v^{\parent}_\bound,D) \ltimes \interval(w)$ \tcc*[r]{computing $(\Join_{F \in \edges_{\nodes^{t}}} R_F)$ via box decomposition}}{
					
					\If{Algorithm~\ref{algo:answer} on $t$ with $v^t_\bound = \pi_{\bag_{\parent} \cap \bag_{t}} (k)$ is empty for all $k$ \label{line:check2} }{
						$\dict_{t}(w, \pi_{\nodes^{\parent}_\bound}(k)) = 0$
					}
					
				}
			}
		}
	}

	%
	\caption{Modifying $(\dict)_{t \in V(\htree)}$ }
	\label{algo:fulljoin}
\end{algorithm}

\smallskip
\introparagraph{(ii) Modify $\dict_{t}$ using semijoins} Algorithm~\ref{algo:fulljoin} shows the construction of the modified dictionary $\dict_t$ to incorporate the semijoin result. The goal of this step is to ensure that if $\dict_t(w, v^{t}_\bound) = 1$, then there exists a set of valid valuations for all variables in the subtree rooted at $t$. We will apply a sequence of semijoin operations in a bottom up fashion which we describe next.

\smallskip
\noindent \introparagraph{Bottom Up Semijoin}  In this phase, the bags are processed according to \textit{post-order} traversal of the tree in bottom-up fashion. The key idea is to stream over all heavy valuations of a node in $\tree_t$ and ensure that they join with some tuple in the child bags. Let $Q_{t}(v^t_\bound, D)$ denote the \textsf{NPRR} join instance on the relations covering variables $\nodes^{t}$ where bound variables are fixed to $v^t_\bound$. When processing a non-root (or non-child of root) node $t_j$, a semijoin is performed with its parent $t_i$ to flip all dictionary entries of $t_i$ from $\textsf{1}$ to $\textsf{0}$ if the entry does not join with any tuple in $t_j$ on their common attributes $\bag_{t_i} \cap \bag_{t_j}$. To perform this operation, we stream over all tuples $k \leftarrow Q_{t_i}(v^{t_i}_\bound, D)$ and check if $\pi_{\bag_{t_i} \cap \bag_{t_j}} (k)$ is present in the join output of relations covering $t_j$. This check in bag $t_j$ can be performed by invoking Algorithm~\ref{algo:answer} with bound valuation $\pi_{\bag_{t_i} \cap \bag_{t_j}} (k)$ in time $\tilde{O}(|D|^{\delta_{t_j}})$.

\smallskip
\noindent \introparagraph{Algorithm Analysis} We will show that Algorithm~\ref{algo:fulljoin} can be executed in time $\tilde{O}(|D|^{\bu^* + \max_t \delta_{t}})$. 

\begin{proposition}
	Algorithm~\ref{algo:fulljoin} executes in time $\tilde{O}(|D|^{\bu^*+ \max_t \delta_{t}})$.
\end{proposition}
\begin{proof}
	The main observation is that the join $Q_{\textsf{parent}}(v^t_\bound,D) \ltimes \interval(w)$ for all nodes at level $L$ in $\tree_{\textsf{parent}}$ can be computed in time at most $O(|D|^{\bu^*})$ as shown in Proposition~\ref{prop:boxmerging}. Since the operation in Line~\ref{line:check2} can be performed in time  $\tilde{O}(|D|^{\delta_{t}})$ for each $k$ and $\tree_\textsf{parent}$ has at most logarithmic number of levels, the total overhead of the procedure is dominated by the semijoin operation where delay for bag $t$ is largest. This gives us the running time of $\tilde{O}(|D|^{\bu^* + \max_t \delta_{t}})$ for the procedure.
\end{proof}

\begin{proposition}
	If $\dict_{t}(w, v^{t}_\bound) = 1$, then there exists a set of valuations for all variables in the subtree rooted at $t$ for $v^{t}_\bound$.
\end{proposition}
\begin{proof}
	Consider a valuation such that $\dict_{t}(w, v^{t}_\bound)  = 1$. If $v^{t}_\bound$ does not join with the relations of any child bag $c$, then Line~\ref{line:check2} would be true and Algorithm~\ref{algo:fulljoin} would have flipped the dictionary entry to $0$. Thus, there exists a valuation for $\nodes^{c}_\free$. Applying the same reasoning inductively to each child bag $c$ till we reach the leaf nodes gives us the desired result.
\end{proof}


The modified dictionary, along with the enumeration algorithm, will guarantee that valuation of free variables that is output by Algorithm~\ref{algo:answer} for a particular bag will also produce an output for the entire query. Note that the main memory requirement of Algorithm~\ref{algo:fulljoin} is only $O(1)$ pointers and the data structures for each bag which takes $\tilde{O}(|D| + |D|^{f})$ space. 

\section{Answering using Query Decomposition}
\label{sec:answering:decomposition}

In this section, we present the query answering algorithm for the given $\nodes_{\bound}$-connex decomposition. We will first add some metadata to each bag in the decomposition and then invoke algorithm~\ref{algo:answer} for each bag in pre-order fashion.

\smallskip
\introparagraph{Adding pointers for each bag} Consider the decomposition $\htree$ along with $(\tree, \dict)_{t \in V(\htree) \setminus t_\bound}$ for each bag. We will modify $\htree$ as follows: for each node of the tree, we fix a pointer $predecessor(t)$, that will point to the {\em pre-order predecessor} of the node. Intuitively, pre-order predecessor of a node is the last node where valuation for a free variable will be fixed in the pre-order traversal of the tree just before visiting the current node. Figure~\ref{fig:queryanswering:decompostion} shows an example of a modified decomposition. This transformation can be done in $O(1)$ time.

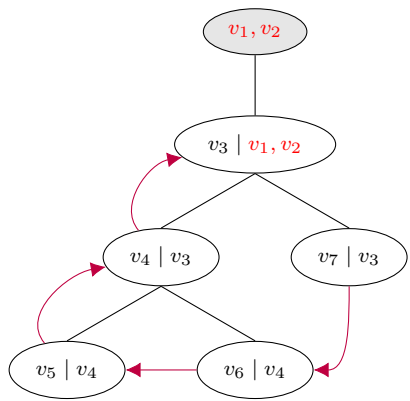
\begin{figure}[htp!]
	\centering
	\begin{tikzpicture}
	[scale=.5,auto,every node/.style={ellipse,draw},level distance=3cm]
	\tikzstyle{level 1}=[sibling distance = 6cm]
	\tikzstyle{level 2}=[sibling distance = 5cm] 
	
	\node [fill=black!10] {\scriptsize $\color{red} v_1, v_2$ }
	child {node (1) {\scriptsize $\prq{ v_3}{\color{red} v_1, v_2}$}
		child {node  (2) {\scriptsize $\prq{ v_4}{ v_3}$}
			child {node (3)  {\scriptsize $\prq{ v_5}{ v_4}$ }}
			child {node (4) {\scriptsize $\prq{ v_6}{ v_4}$ }}
		}
		child {node (5) {\scriptsize $\prq{ v_7}{ v_3}$}}};
	
	\draw[purple, ->] (2) to[ out = 130, in = 190] (1);
	\draw[purple, ->] (3) to[ out = 130, in = 190] (2);
	\draw[purple, ->] (4) to[ out = 180, in = 0] (3);
	\draw[purple, ->] (5) to[ out = 270, in = 0] (4);
	\end{tikzpicture}
	\caption{Example of the modified tree decomposition: the arrows in color are the predecessor pointers} \label{fig:queryanswering:decompostion}
\end{figure}

For ease of description of the algorithm, we assume that the Algorithm~\ref{algo:answer} answering $Q^{\eta}[v^{t}_{\bound}]$ for any bag $t$ is accessible using the procedure $\textsf{next}_t(v^{t}_{\bound})$. Let $\nodes^{t}_{\textsf{pred}}$ represent all bound variables encountered in the pre-order traversal of the tree from $t_\bound$ to $t$ (including bound variables of $t$).

\smallskip
\introparagraph{Algorithm Description} The algorithm begins from the root node and fixes the valuation for all free variables in root bag. Then, it proceeds to the next bag recursively considering all ancestor variables as bound variables and finds a valuation for $\bag_{t} \setminus \textsf{anc}(t)$. At the first visit to any bag, if the bound variables $v^{t}_{\bound}$ do not produce an output in delay $\tilde{O}(|D|^{\delta_{t}})$, then we proceed to the next valuation in the parent bag. However, if the enumeration for some $v^{t}_{\bound}$ did produce output tuples but the procedure $\textsf{next}_t(v^{t}_{\bound})$ has finished, we proceed to the predecessor of the bag to fix the next valuation for variables in predecessor bag. In other words, the ancestor variables remain fixed and we enumerate the cartesian product of the remaining variables.

\begin{lemma}
	Algorithm~\ref{algo:answering:decomposition} enumerates the answers with delay at most $\tilde{O}(|D|^{h})$ where $h$ is the $\delta$-height of the decomposition tree. Moreover, it requires at most $O(\log |D|)$ memory
\end{lemma}
\begin{proof}
	Since the size of the decomposition is a constant, we require at most $O(1)$ pointers for predecessors and $O(1)$ pointers for storing the valuations of each free variable. Let $n_\ell$ be the set of nodes at depth $\ell$. We will express the delay of the algorithm in terms of the delay of the subtrees of every node. The delay at the root $t_\bound$ after checking whether valuation $v_\bound$ is in the base relations is $d_{t_\bound} = O(\sum_{t \in n_1} d_{t})$.  This is because the enumeration of each subtree rooted at depth $\ell = 1$ depends only on its ancestor variables and is thus independent of the other subtrees at that depth. Since each node in the tree can produce at most $|D|^{\delta_{t}}$ valuations in $\tilde{O}(|D|^{\delta_{t}})$ time, the recursive expansion of $\delta_{t_\bound}$ gives $\delta_{t_\bound} = O( \sum_{p \in \lambda(\htree)} \tilde{O}(|D|^{\sum_{t \in p} \delta_{t} }))$. The largest term over all root to leaf paths is $\tilde{O}(|D|^{h})$ which gives us the desired delay guarantee.
\end{proof}

\begin{algorithm}[!htp]
	\SetCommentSty{textsf}
	\DontPrintSemicolon 
	\SetKwFunction{proc}{\textsf{eval}}
	\SetKw{KwGoTo}{go to}	
	\SetKwInOut{Input}{\textsc{input}}\SetKwInOut{Output}{\textsc{output}}
	\Input{tree $\htree$, $(\tree, \dict)_{t \in V(\htree)}, v_b$}
	\Output{query answer $Q(D)$}
	\BlankLine
	Initialize $t_{visited} \leftarrow 0$ for all nodes, $v \leftarrow v_b, t \leftarrow$ left child of $t_\bound$, $parent(t) \leftarrow t$ \\
	Check if $R_F(v_\bound) \neq \emptyset, F \in \edges, F \subseteq C$ \\ 
	\ForAll{nodes in pre-order traversal starting from $t$}{
		$v \leftarrow \pi_{\nodes^{t}_{\textsf{pred}}}(v)$ 		\label{loop:begin}  \\
		$v^{t}_\free \leftarrow \textsf{next}_t( \pi_{\nodes^{t}_\bound} (v))$ \\
		\If{$v^{t}_\free$ is empty and $t_{visited} = 0$}{
			$t \leftarrow parent(t)$ \\
			$continue$
		}
		\If{$v^{t}_\free$ is empty and $t_{visited} = 1$}{
			$t_{visited} \leftarrow 0$ \\
			$t \leftarrow predecessor(t)$ \\
			$continue$
		}
		$t_{visited} \leftarrow 1$ \\
		$v \leftarrow (v, v^{t}_\free)$ \\
		\If{$t$ is last node in the tree}{
			\If{$R_F(v) \neq \emptyset, F \in \edges$}{
				emit {$v$} \\
			}
			\KwGoTo line~\ref{loop:begin} \tcc*[r]{If $t$ is last node in tree, find next valuation for $\nodes^{t}_{\free}$}
		}
	}	
	\BlankLine
	\caption{Query Answering using $\nodes_{\bound}$-connex decomposition}
	\label{algo:answering:decomposition}
\end{algorithm}

\section{Comparing width notions} \label{sec:width}

We briefly discuss the connection of $\textsf{fhw}(\hgraph \mid \nodes_\bound)$ for a $\nodes_{\bound}$-connex decomposition with other related hypergraph notions. 

The first observation is that the minimum edge cover number $\rho^*$ is always an upper bound on $\textsf{fhw}(\hgraph \mid \nodes_\bound)$. 
On the other hand, the $\textsf{fhw}(\hgraph \mid \nodes_\bound)$ is incomparable with $\fhw{\hgraph}$.
Indeed, Example~\ref{ex:cbound} shows that $\textsf{fhw}(\hgraph \mid \nodes_\bound)  < \fhw{\hgraph}$, and the
example below shows that the inverse situation can happen as well.

\begin{example}
	The query $R(x,y), S(y,z)$ is acyclic and has $\fhw{\hgraph} = 1$. Let
	$\nodes_{\bound} = \{x,z\}$. The only valid $\nodes_{\bound}$-bound decomposition is the one with two bags,
	$\{x,z\}, \{x,y,z\}$, and hence $\textsf{fhw}(\hgraph \mid \nodes_\bound) = 2$. In this scenario,
	$\textsf{fhw}(\hgraph \mid \nodes_\bound)  > \fhw{\hgraph}$. \label{ex:first}
\end{example} 

\begin{example} \label{ex:cbound}
	Figure~\ref{fig:ex:fhw} shows an example hypergraph and a $\nodes_{\bound}$-bound tree decomposition
	(the variables in $\nodes_{\bound}$ are colored red). 
	For this example, $\textsf{fhw}(\hgraph) = 2$, but $\textsf{fhw}(\hgraph \mid \nodes_{\bound}) = 3/2$. Indeed, observe that we can cover the lower bag of the tree decomposition with a 
	fractional edge cover of value only $3/2$.
\end{example}

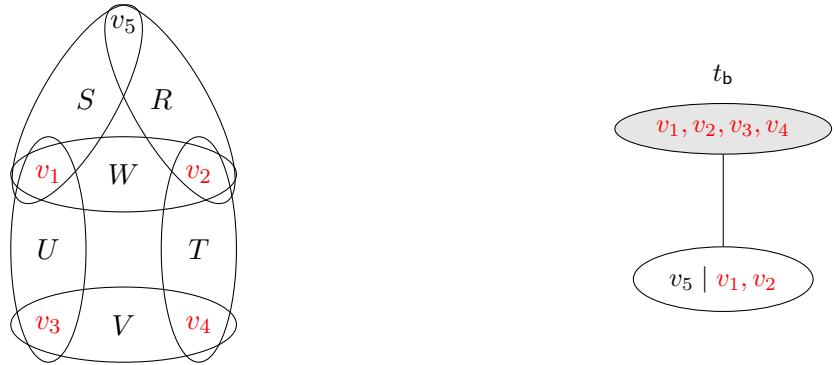
\begin{figure}[!htp]
	\begin{minipage}{0.60\linewidth}
		\centering
		\begin{tikzpicture}
		\tikzset{edge/.style = {->,> = latex'}}
		\begin{scope}[fill opacity=1]
		
		\draw[] (0,0) ellipse (1.5cm and 0.5cm);
		\draw[] (0,-2) ellipse (1.5cm and 0.5cm);
		\draw[rotate=90] (-1,-1) ellipse (1.5cm and 0.5cm);
		\draw[rotate=90] (-1,1) ellipse (1.5cm and 0.5cm);
		
		\draw[rotate=60] (0.5,1) ellipse (1.5cm and 0.5cm);
		\draw[rotate=300] (-0.5,1) ellipse (1.5cm and 0.5cm);	
		\node (E) at (0,2) {$v_5$};
		\node (A) at (-1,0) {$\color{red} v_1$};
		\node (B) at (1,0) {$\color{red} v_2$};
		\node (F) at (0,0) {$W$};
		\node (G) at (0,-2) {$V$};
		\node (H) at (-1,-1) {$U$};
		\node (I) at (1,-1) {$T$};
		\node (J) at (0.5,1) {$R$};				
		\node (J) at (-0.5,1) {$S$};					
		\node (C) at (-1,-2) {$\color{red} v_3$};
		\node (D) at (1,-2) {$\color{red} v_4$};
		\end{scope}	
		\end{tikzpicture}
	\end{minipage}
	\begin{minipage}{0.35\linewidth}
		\centering
		\begin{tikzpicture}
		[scale=.5,auto,every node/.style={ellipse,draw},level distance=4cm]
		\tikzstyle{level 1}=[sibling distance = 4cm]
		\tikzstyle{level 2}=[sibling distance = 4cm] 
		\node [fill=black!10, label = {\small $t_\bound$}] {\small $\color{red} v_1, v_2, v_3, v_4$ }
		child {node (1) {\small $\prq{ v_5}{\color{red} v_1, v_2}$}};
		
		\end{tikzpicture}
	\end{minipage}
	\caption{Query hypergraph and corresponding $C$-bound tree decomposition 
		with $C = \{v_1, v_2, v_3, v_4\}$}
	\label{fig:ex:fhw}
\end{figure}

\end{document}